\documentclass[lettersize,journal]{IEEEtran}
\usepackage{amsmath,amsfonts}
\usepackage{algorithmic}
\usepackage{algorithm}
\usepackage{array}
\usepackage{textcomp}
\usepackage{stfloats}
\usepackage{url}
\usepackage{verbatim}
\usepackage{graphicx}
\usepackage{cite}
\usepackage{amsmath,amsthm}   
\usepackage{amssymb}
\usepackage{nomencl}
\usepackage{multirow}
\usepackage{bm}
\usepackage{lipsum}
\usepackage{makeidx}
\usepackage{enumerate}
\usepackage{color}
\usepackage[font=small,labelfont=bf]{caption}
\usepackage[font=small,labelfont=bf]{subcaption}
\usepackage{tikz,pgfplots}
\usetikzlibrary{calc}
\usetikzlibrary{intersections}
\usetikzlibrary{arrows,shapes}
\usetikzlibrary{positioning}

\newtheorem{corollary}{Corollary}

\newtheorem{lemma}{Lemma}

\theoremstyle{definition}

\allowdisplaybreaks

\begin{document}
\title{Rate Splitting in MIMO RIS-assisted Systems with Hardware Impairments and Improper Signaling}

\author{Mohammad Soleymani$^*$, 
Ignacio Santamaria$^\dag$ \emph{Senior Member, IEEE}, and
Eduard Jorswieck$^\ddag$ \emph{Fellow, IEEE}
 \\ \thanks{ 
$^*$Mohammad Soleymani is with the Signal and System Theory Group, Universit\"at Paderborn, , 33100 Paderborn, Germany   
(e-mail: \protect\url{mohammad.soleymani@sst.upb.de}).  

$^\dag$Ignacio Santamaria is with the Department of Communications Engineering, University of Cantabria, 39005 Santander, Spain (e-mail: \protect\url{i.santamaria@unican.es}).

$^\ddag$ Eduard Jorswieck is with the Institute for Communications Technology, Technische Universit\"at Braunschweig, 38106 Braunschweig, Germany
(e-mail: \protect\url{jorswieck@ifn.ing.tu-bs.de})

The work of Ignacio Santamaria has been partly  supported by the project ADELE PID2019-104958RB-C43, funded by MCIN/ AEI
/10.13039/501100011033.

The work of Eduard Jorswieck was supported in part by the Federal Ministry of Education and Research (BMBF, Germany) in the program of ``Souver\"an. Digital. Vernetzt.'' joint project 6G-RIC, project identification number: 16KISK020K and 16KISK031.
}}
\maketitle
\begin{abstract}
In this paper, we propose an optimization framework for rate splitting (RS) techniques in multiple-input multiple-output (MIMO) reconfigurable intelligent surface (RIS)-assisted systems, possibly with I/Q imbalance (IQI). This framework can be applied to any optimization problem in which the objective and/or constraints are linear functions of the rates and/or transmit covariance matrices. Such problems include minimum-weighted and weighted-sum rate maximization, total power minimization for a target rate, minimum-weighted energy efficiency (EE) and global EE maximization.  The framework may be applied to any interference-limited system with hardware impairments. For the sake of illustration, we consider a multicell MIMO RIS-assisted broadcast channel  (BC) in which the base stations (BSs) and/or the users may suffer from IQI. Since IQI generates improper noise, we consider improper Gaussian signaling (IGS) as an interference-management technique that can additionally compensate for IQI. We show that RS when combined with IGS can substantially improve the spectral and energy efficiency of overloaded networks (i.e., when the number of users per cell is larger than the number of transmit/receive antennas).
\end{abstract} 
\begin{IEEEkeywords}
 Energy efficiency,  improper Gaussian signaling, majorization minimization, MIMO broadcast channels, power minimization, rate splitting, reflecting intelligent surface, spectral efficiency.
\end{IEEEkeywords}

\section{Introduction}
Rate splitting (RS) and reconfigurable intelligent surfaces (RISs) are among  the most promising technologies for beyond 5G (B5G) and 6G, and have been shown to be able to substantially improve the spectral and energy efficiency of various wireless communication systems \cite{mao2022rate,     wu2021intelligent, di2020smart}.
RS is a powerful interference-management technique, which includes a variety of schemes/technologies such as treating interference as noise (TIN), non-orthogonal multiple access (NOMA) techniques, space division multiple access (SDMA), multicasting, broadcasting, among others \cite{mao2022rate}.  
Additionally, the use of RIS is an emerging trend in wireless technologies for improving the coverage and/or manage/neutralize interference by modulating channels \cite{wu2021intelligent}. 

In this paper, we provide a general optimization framework for RS in multiple-input, multiple-output (MIMO) RIS-assisted systems and investigate the performance of RS in such systems. 

\subsection{Literature review}
One of the main bottlenecks for modern wireless communication systems is interference from other users such that these systems are mostly interference-limited \cite{andrews2014will}. Hence, interference-management techniques are expected to continue playing an essential role in the upcoming wireless communication systems. RS is a practical and flexible  signaling scheme that includes various schemes/technologies such as TIN, NOMA, SDMA, multicasting, broadcasting, etc. \cite{mao2022rate}. Indeed, RS can get the benefits of all theses schemes and switch between them depending on the channel conditions and interference level. 
When interference is weak, TIN is the optimal decoding strategy for maximizing the sum rate \cite{annapureddy2009gaussian}. Under some conditions on the strength of the desired and interference links, TIN is also optimal in terms of  the generalized degrees of freedom \cite{ geng2015optimality}. In the presence of strong interference, decoding and canceling interference from the received signal is the optimal strategy \cite{sato1981capacity}, which is also known as successive interference cancellation (SIC). RS bridges both strategies as it may apply each depending on the level of interference. In RS, there are two types of messages: common and private. Common messages are decoded by all users while treating the private messages as noise. However, the private messages are decoded only by the intended user, employing SIC to remove common messages from the received signal. It is worth noting that for the operational points between these two extreme cases, i.e., weak and strong interference, the optimal strategy for every interference-limited system is not known. In other words, RS is not necessarily the optimal transmission/decoding strategy for all operational points \cite{el2011network}. 

 RS has received a lot of attention in the past few years  \cite{mao2022rate, hao2015rate, lu2017mmse, clerckx2019rate,  zhou2021rate, li2020rate,  mao2021rate,
mao2020beyond,flores2020linear,li2021linearly,flores2021tomlinson, dizdar2021rate,
 yang2012degrees, joudeh2016robust, hao2017achievable, mao2018rate, bansal2021rate, li2022rate}; however, it is not a new technology. The terminology of ``RS multiple access'' (RSMA) was introduced in \cite{rimoldi1996rate} for the first time in the literature, where the authors showed that RS may enlarge the rate region of a single-input, single-output (SISO) multiple-access channel (MAC). The main idea of RS is even older and was introduced by Carleial in \cite{carleial1978interference} for the 2-user interference channel (IC), where it was shown that RS can enlarge the achievable rate region.
A survey on RS is provided in \cite{mao2022rate}, and we refer the reader to  \cite[Sec. II]{mao2022rate} for a detailed literature review on RSMA. In this section, we only briefly discuss some advantages of RSMA. Note that RSMA is more general than NOMA  and, in fact, includes NOMA as a particular case. In NOMA, users are firstly ordered, and then, each user employs SIC to decode and cancel the signal of the users with a lower order. Hence,  user ordering plays a key role in NOMA, and finding the optimal user ordering can be a very difficult task in RIS-assisted MIMO systems. NOMA is optimal in SISO systems with perfect channel state information (CSI) and without RIS \cite{el2011network}. The channels are degraded  in SISO systems, which makes it easier to obtain the optimal user ordering. However, MIMO channels are not degraded and hence, NOMA is suboptimal  for MIMO systems and might be very inefficient  \cite{clerckx2021noma}. The optimal user ordering can be very challenging in multiple-antenna RIS-assisted systems and may even require solving an NP-hard problem \cite{ni2021resource}. 
It is also in general infeasible to obtain the optimal user ordering through exhaustive search since there are $K!$ user ordering possibilities. Furthermore, the optimal user ordering problem may be further complicated in the presence of imperfect CSI. Unlike NOMA, RSMA does not need any user ordering and has been shown to be robust against imperfect CSI \cite{yang2012degrees, joudeh2016robust}. The performance of RSMA in MIMO systems has to be developed further, but the existing literature shows that RSMA may improve the spectral efficiency of 2-user MIMO IC and/or BC \cite{hao2017achievable, li2021linearly}. To summarize, RSMA can be adapted to the interference level and is robust against imperfect CSI. Additionally, RSMA can be very efficient in multiple-antenna systems. 

Another promising technology for 6G is RIS, which has been shown to  significantly improve the spectral and energy efficiency of various wireless communication systems \cite{wu2021intelligent, di2020smart, huang2020holographic, huang2019reconfigurable, wu2019intelligent, kammoun2020asymptotic, yu2020joint,  yang2020risofdm, pan2020multicell, zhang2020intelligent, elmossallamy2020reconfigurable, ni2021resource, wei2021channel, zhou2020framework}. 
RIS can modulate the channels, thus providing another degree-of-freedom to improve the coverage and/or manage interference. The papers \cite{huang2019reconfigurable, kammoun2020asymptotic, yu2020joint, pan2020multicell, yu2020improper, soleymani2022improper, soleymani2022noma} showed that RIS can improve the performance of single-cell and/or multi-cell BCs. The paper \cite{jiang2021achievable, huang2020achievable} showed that RIS can enlarge the rate region of the $K$-user multiple-input, single-output (MISO) ICs. The authors in \cite{jiang2022interference} have employed RIS to neutralize interference in the $K$-user SISO IC. In \cite{li2021intelligent}, it is shown that RIS can improve the performance of orthogonal-frequency-division-multiple-access (OFDMA) systems. We refer the reader to \cite{wu2021intelligent, di2020smart} for a more detailed literature review on RIS.

In addition to RS and RIS, there are other interference-management tools such as improper Gaussian signaling (IGS) \cite{cadambe2010interference, javed2018improper, soleymani2020improper, gaafar2017underlay, amin2017overlay, soleymani2019energy, lameiro2015benefits, lameiro2017rate, lagen2016superiority, lagen2016coexisting,  zeng2013transmit, ho2012improper, soleymani2019improper,  soleymani2019robust, soleymani2019ergodic, Sole1909:Energy, tuan2019non, nasir2020signal, nasir2019improper, nguyen2021improper, yu2020improper, yu2020improper2, park2013sinr, yu2021maximizing, javed2020journey}. In a zero-mean proper complex Gaussian signal, the real and imaginary parts of the signal are independent and identically distributed (iid). However, the real and imaginary parts  of improper signals can be correlated and/or have unequal powers \cite{schreier2010statistical}. Indeed, relaxing the assumption that the transmit signals  are proper, it is possible to exploit another degree-of-freedom in the design by considering the real and imaginary parts of each transmitted signal as independent optimization dimensions. 
IGS has been shown to increase the DoF of the 3-user SISO IC for the first time in the literature in \cite{cadambe2010interference}. Later, IGS was employed to improve the spectral and/or energy efficiency of various interference-limited systems such as multi-user ICs \cite{zeng2013transmit, ho2012improper, soleymani2019improper,  soleymani2019robust, soleymani2019ergodic, Sole1909:Energy}, cognitive radio \cite{gaafar2017underlay, amin2017overlay, soleymani2019energy, lameiro2015benefits}, BCs \cite{yu2020improper, soleymani2022improper, soleymani2022noma}, among others. We refer the reader to \cite{javed2020journey} for a more comprehensive literature review on improper signaling in both Gaussian signals and discrete constellations.

Interference is not, unfortunately, the only factor limiting the performance of wireless systems. There are also other factors that limit performance such as hardware impairments (HWI). The papers \cite{hao2015rate, lu2017mmse, clerckx2019rate, mao2020beyond, flores2020linear, zhou2021rate, li2020rate, mao2021rate, flores2021tomlinson, dizdar2021rate, carleial1978interference, han1981new, rimoldi1996rate} considered the performance of RS with ideal devices. However, in practice, devices always have non-idealities that, if not adequately compensated for, can severely affect the system performance  \cite{boshkovska2018power,  soleymani2019improper, soleymani2020improper, soleymani2022improper, soleymani2020rate, javed2018improper, javed2019multiple, boulogeorgos2016energy, papazafeiropoulos2017rate}. 
A source of imperfection in devices is I/Q imbalance (IQI), caused  by an amplitude and phase mismatch between the in-phase and quadrature components \cite{javed2019multiple, soleymani2020improper, boulogeorgos2016energy}. IQI is modeled as a widely linear transformation  of the input signal, which makes the output  (transmitted) signal improper \cite{javed2019multiple}.  A way to compensate for the IQI is to employ IGS, which further motivates us to study IGS in this work \cite{soleymani2020improper, soleymani2022improper}.

\begin{table}
\centering
\footnotesize
\caption{A brief comparison of the most related works.}\label{table-1}
\begin{tabular}{|c|c|c|c|c|c|c|c|c|c|c|c|c|c|c|}
	\hline
&IGS&RS&RIS&MIMO&STAR-RIS&HWI
 \\
\hline
  This paper&$\surd$&$\surd$&$\surd$&$\surd$&$\surd$&$\surd$\\
\hline
   \cite{soleymani2022improper}&$\surd$&&$\surd$&$\surd$&&$\surd$\\
\hline
    \cite{soleymani2022noma, yu2020joint, yu2021maximizing}
&$\surd$&&$\surd$&&&
\\
\hline
\cite{javed2018improper, soleymani2019improper}&$\surd$&&&&&$\surd$
\\
\hline
\cite{ho2012improper,   soleymani2019robust, soleymani2019ergodic} 
&$\surd$&&&&&
\\
\hline
\cite{soleymani2020improper}&$\surd$&&&$\surd$&&$\surd$
\\
\hline
\cite{lagen2016superiority, lagen2016coexisting}
&$\surd$&&&$\surd$&&
\\
\hline
\cite{ huang2019reconfigurable, kammoun2020asymptotic, wu2019intelligent}
&&&$\surd$&&&
\\
\hline
\cite{ pan2020multicell, zhang2020intelligent}
&&&$\surd$&$\surd$&&
\\
\hline
\cite{hao2015rate, lu2017mmse, clerckx2019rate,  zhou2021rate, li2020rate,  mao2021rate }
&&$\surd$&&&&
\\
\hline
\cite{ mao2020beyond,flores2020linear,li2021linearly,flores2021tomlinson, dizdar2021rate}
&&$\surd$&&$\surd$&&
\\
\hline
\cite{bansal2021rate,li2022rate}
&&$\surd$&$\surd$&&&
\\
\hline
\cite{papazafeiropoulos2017rate}&&$\surd$&&&&$\surd$
\\
\hline
\cite{mu2021simultaneously, wu2021coverage, liu2021star, xu2021star}
&&&$\surd$&&$\surd$&
\\
\hline
		\end{tabular}
\normalsize
\end{table} 

\subsection{Motivation}
In Table \ref{table-1}, we provide a brief comparison of some of the most related works, based on the considered scenarios.  As can be observed in the table, even though RS is not a new concept/technology, further studies are needed to analyze its performance, especially in MIMO RIS-assisted systems and/or in the presence of imperfect devices. Additionally, it may be interesting to evaluate the performance of RS when used in combination with another interference-management technique such as IGS.  Despite the vast recent literature on the topic, we believe this is the first work to study RS, jointly with IGS, in MIMO RIS-assisted systems.

IGS is  a powerful interference-management tool, which can be beneficial in RSMA  systems.
In our previous study \cite{soleymani2022noma}, we have proposed a NOMA-based IGS scheme for a multicell MISO RIS-assisted BC showing that IGS and NOMA can be mutually beneficial interference-management techniques. As indicated, NOMA is suboptimal in MIMO systems. 
Furthermore, NOMA involves high complexities to obtain the optimal user ordering in multicell and/or RIS-assisted systems, even in a MISO scenario, which makes it very complicated to implement a full NOMA  scheme among the cell users. Thus, it can be challenging (and inefficient) to extend the scheme in \cite{soleymani2022noma} to MIMO systems. 
In \cite{soleymani2022improper}, we proposed IGS schemes for multicell MIMO RIS-assisted BCs with TIN. Unfortunately, TIN can be highly suboptimal in strong interference. Thus, the performance of IGS can be highly improved if combined with another interference-management technique, especially in overloaded systems as shown in \cite{soleymani2022noma}. 
These  previous findings motivate us to study schemes with RS and IGS for MIMO RIS-assisted systems with HWI. 

Another motivation for this work is to propose a general framework to solve complicated optimization problems for RSMA in MIMO RIS-assisted systems with HWI. Such a framework could  be applied to any interference-limited systems and include different RS schemes. 
\subsection{Contributions}
In this paper, we propose a general framework for RSMA in MIMO RIS-assisted systems to solve a rich class of optimization problems in which the objective function and/or constraints are linear functions of the rates and/or transmit covariance matrices. These optimization problems include weighted-minimum-rate and weighted-sum-rate maximization, power minimization for a target rate, globally EE (GEE) maximization, and minimum-weighted-EE maximization, among others that are vastly employed in the literature \cite{zeng2013transmit, Sole1909:Energy, tuan2019non, nasir2020signal, nasir2019improper, yu2020improper, soleymani2021distributed}. The optimization framework employs majorization minimization (MM) combined with alternating optimization (AO) to optimize the transmit covariance matrices and RIS components.  We first introduce the framework and then, present some examples that illustrate how it can be specialized to solve different optimization problems.   

We make realistic assumptions regarding devices and assume that the transceivers may suffer from IQI, based on the model in \cite{javed2019multiple}. 
Additionally, we study practical scenarios for RIS by employing an appropriate model for the large-scale and small-scale fading of the links as well as considering four different feasibility sets for RIS components based on \cite{wu2021intelligent}. Furthermore, we consider simultaneous transmit and reflect (STAR) RISs and show that STAR-RIS can considerably outperform regular RIS. Supported by  experimental results \cite{docomo2020docomo, zhang2021dynamical, wang2018simultaneous}, STAR-RIS is able not only to reflect signals, but also to transmit.
As shown in this work, in some scenarios, STAR-RIS can improve the spectral efficiency of the network, outperforming regular RIS.

The numerical results show that IGS with RS can significantly increase the minimum rate and EE of users for a given power budget, or reduce the total power consumption for a given target rate in overloaded systems, i.e., when the number of users per cell is higher than the number of transmit/receive antennas.  Additionally, our results show that RS and RIS are mutually beneficial tools  to improve performance in overloaded systems. In this case, RS aims at managing interference, while RIS aims at improving the coverage, and they complement each other very well. Furthermore, our results indicate that IQI may considerably decrease the system performance if we do not account for it in the system design.  

The main contributions of this paper are summarized as follows:
\begin{itemize}
\item We propose a general framework for RSMA, which can be applied to MIMO RIS-assisted interference-limited networks. This framework is capable of solving any optimization problem in which the objective and/or constraints are linear functions of the rates and/or transmit covariance matrices. 
 
\item For the sake of illustration, we consider the downlink of a multicell MIMO RIS-assisted BC. Then, we formulate and solve a variety of optimization problems for the system, including weighted-minimum rate and weighted-sum rate maximization, weighted-minimum EE and global EE maximization, power minimization for a target rate.

\item We consider realistic assumptions regarding devices and RISs. We additionally extend our framework to MIMO STAR-RIS assisted systems. To the best of our knowledge, this is the first work on RS in STAR-RIS assisted systems.

\item We show that RS with IGS can substantially improve the spectral and energy efficiency of overloaded interference-limited systems. The more overloaded the system is, the more benefits are provided by RS/IGS. In other words, the benefits of RS/IGS increase with number of users per cell and decrease with the number transmit/receive antennas.  
\end{itemize}

\subsection{Paper outline}
This paper is organized as follows. Section \ref{sec-sys} describes the system model and formulates the optimization problem addressed in this work. 
Section \ref{sec-cov-opt} presents the solution for optimizing the transmit covariance matrices. 
Section \ref{sec=opt-ris} extends the framework to MIMO RIS-assisted systems. 
Section \ref{sec-num} presents some numerical results, and finally, Section \ref{sec-con} concludes the paper. Additionally, we provide some  preliminaries on improper signaling,  majorization minimization, and RIS in appendices.

\section{System model}\label{sec-sys}
\begin{figure}[t!]
    \centering
\includegraphics[width=.4\textwidth]{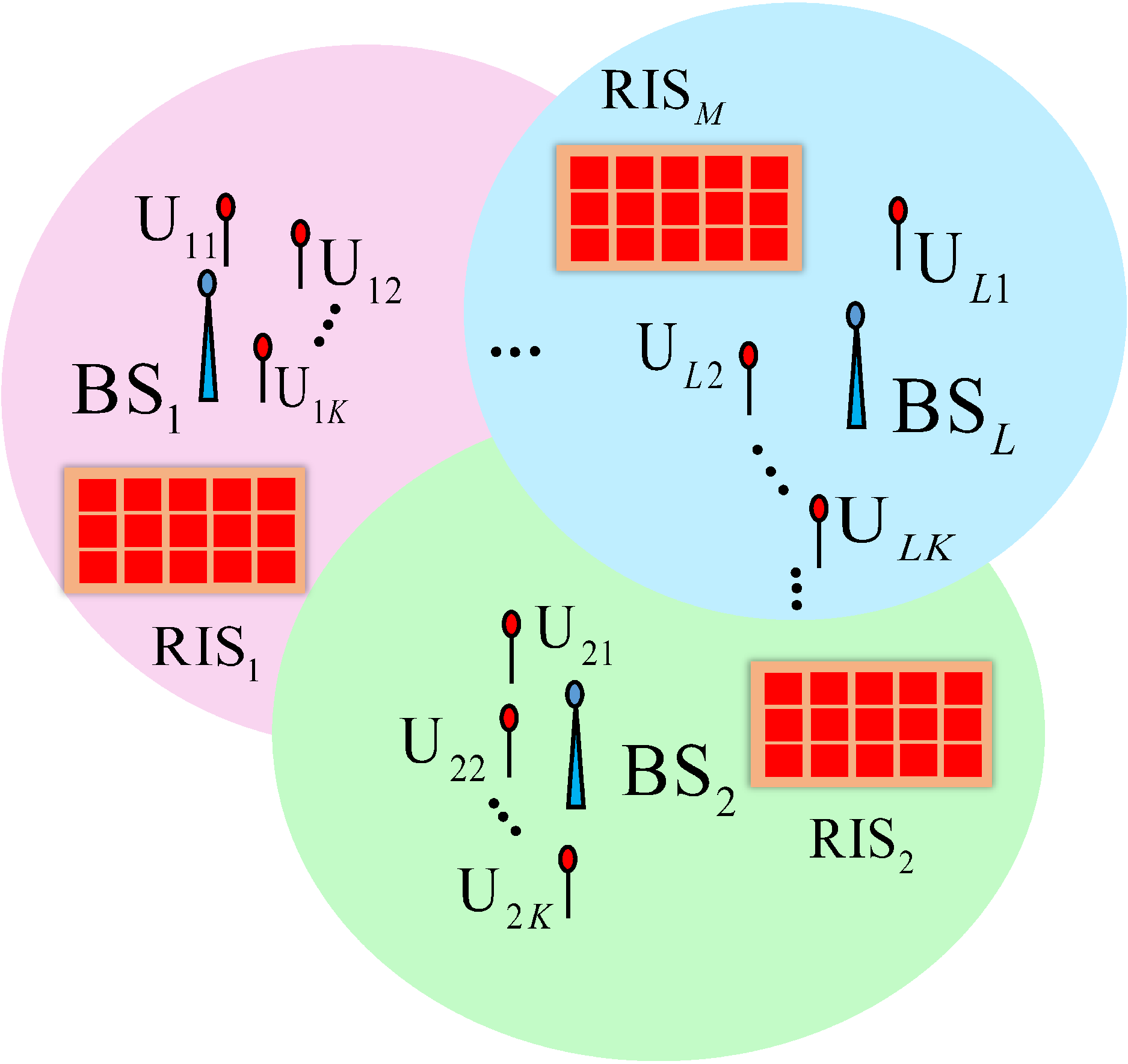}
     \caption{A multicell broadcast channel with RIS.}
	\label{Fig-sys-model}
\end{figure}
Our proposed framework can be applied to any MIMO RIS-assisted interference-limited system with IQI at the transceivers. As a representative example of these networks, we consider a  multicell MIMO BC with at least one RIS per cell shown in Fig. \ref{Fig-sys-model}, 
which is among the most practical scenarios \cite{pan2020multicell, shi2011iterative}. 
The system consists of $L$ BSs with $N_{BS}$ antennas, each serving $K$ users with $N_u$ antennas. All BSs and users may suffer from IQI, according to the model in \cite{javed2019multiple}, which is shortly presented in Appendix \ref{sec-iqi}. There are also $M\geq L$ RISs with $N_{RIS}$ components each, which assist the BSs to improve the spectral and/or energy efficiency of the system. 
In \cite{soleymani2022improper}, we showed that a distributed RIS implementation may outperform a collocated RIS implementation due to the harsh large-scale fading of RIS links. Thus, we assume that there is at least one RIS per cell.  To simplify the notation and the presentation of results, we consider a symmetric scenario in which each BS/user/RIS has the same number of antennas/components, and all devices have the same level of IQI. However, our framework can be easily extended to asymmetric systems with any arbitrary number of users per  cell, with different number of antennas at each BS/user/RIS and with different IQI parameters at each device as will be indicated later.

\subsection{RIS model}

In this paper, we employ the RIS model in \cite{pan2020multicell} for the MIMO multicell BC. For the sake of completeness, we briefly restate the model in this subsection and refer the readers to   Appendix \ref{sec=ap=ris} and \cite{wu2021intelligent, di2020smart} for more detailed discussions on the small-scale and large-scale fading models of  RIS-assisted systems, the feasibility sets for RIS components and the format of STAR-RISs. 
In RIS-assisted systems, there are two possible links between a transmitter and receiver: a direct link and a link through RISs. The direct link is given and cannot be optimized. However, the links through RISs can be optimized by modifying the reflecting coefficients. 

The channel matrix between BS $i$ and the $k$-th user associated to BS $l$, denoted as u$_{lk}$, is given by   
\begin{equation}\label{ch-equ}
\mathbf{H}_{lk,i}\!\left(\{\bm{\Theta}\}\right)\!\!=\!\!
\underbrace{\sum_{m=1}^M\mathbf{G}_{lk,m}\bm{\Theta}_m\mathbf{G}_{m,i}}_{\text{Link through RIS}}
+\!\!\!
\underbrace{\mathbf{F}_{lk,i}}_{\text{Direct link}}
\!\!\!\!
\in\mathbb{C}^{N_u\times N_{BS}}
\!\!,
\end{equation}
where $\mathbf{F}_{lk,i}\in\mathbb{C}^{N_u\times N_{BS}}$ is the channel matrix between the BS $i$ and u$_{lk}$,  
$\mathbf{G}_{lk,m}\in\mathbb{C}^{N_u\times N_{RIS}}$ is the channel matrix between the $m$-th RIS and u$_{lk}$, 
$\mathbf{G}_{m,i}\in\mathbb{C}^{N_{RIS}\times N_{BS}}$ is the channel matrix between the BS $i$ and the $m$-th RIS, $\{\bm{\Theta}\}=\{\bm{\Theta}_m\}_{m=1}^M$ denotes the set of all reflecting coefficients, where $\bm{\Theta}_m\in\mathbb{C}^{N_{RIS}\times N_{RIS}}$ is a diagonal matrix, containing the vector of reflecting coefficients for the $m$-th RIS 
\begin{equation*}
\bm{\Theta}_m=\text{diag}\left(\theta_{m1}, \theta_{m2},\cdots,\theta_{m{N_{RIS}}}\right),
\end{equation*}
where $\theta_{mn}$s for all $m,n$ are complex-valued optimization parameters. 
In this way, the channel matrices are functions of the reflecting coefficients. We represent the feasibility set of the reflecting coefficients by $\mathcal{T}$, unless we explicitly refer to one of the feasibility sets in Appendix \ref{sec=ap=ris}. To simplify the representation of equations, we drop the dependency of channels on the reflecting coefficients and denote them as $\mathbf{H}_{lk,i}$ for all $l,k,i$, hereafter. Note that we can easily apply the channel model in \eqref{ch-equ} to asymmetric scenarios in which each BS/user/RIS has a different number of antennas/components. In this case, $\mathbf{H}_{lk,i}\in\mathbb{C}^{N_{u,lk}\times N_{BS,i}}$, where $N_{u,lk}$ and $N_{BS,i}$ are, respectively, the number of antennas at u$_{lk}$ or BS $i$ in an asymmetric scenario. Additionally, note that we consider perfect, instantaneous and global CSI similar to many other studies on RIS such as \cite{ni2021resource, huang2019reconfigurable, wu2019intelligent, kammoun2020asymptotic, pan2020multicell, zhang2020intelligent, zuo2020resource, mu2020exploiting, yang2021reconfigurable, yu2020joint, jiang2022interference}. Such an assumption provides an upper bound for the performance of the proposed techniques while showing potential tradeoffs in the system's performance.
However, in practice, it may happen that we only have an access to noisy estimates of the channels and/or  statistical CSI especially in RIS-assisted multi-user systems \cite{zhou2020framework, wei2021channel}.  Investigating the performance of RS and IGS with imperfect and/or statistical CSI poses numerous challenges that would require a separate study and treatment, which should be considered in a future work.

\subsection{Signal model}
We assume that BSs employ the 1-layer RS scheme to transmit signals to the users.
Note that there are various RS schemes such as 1-layer RS, 2-layer hierarchical RS (HRS), generalized RS \cite[Section III.B]{mao2022rate}. 1-layer RS is not necessarily the optimal RS scheme; however, we consider this scheme since it is very practical and the most widely studied in the literature \cite{hao2015rate, lu2017mmse, clerckx2019rate, mao2020beyond, flores2020linear, zhou2021rate, li2020rate, mao2021rate, flores2021tomlinson, dizdar2021rate}. 
Nevertheless, the proposed framework can be applied to multi-layer or to generalized RS since the rate expressions for all these schemes have similar structure. Due to  space restrictions, we leave multi-layer RS schemes for a future study. 

In the 1-layer RS scheme, the transmit signal of BS $l$ consists of two parts. One part contains a common message, which is decoded by all users associated to that BS. The other part contains private messages that can be decoded only by the intended user. Hence, the transmit signal of  BS $l$ employing 1-layer RS is 
\begin{equation*}
\mathbf{x}_l=\mathbf{x}_{l,c}+\sum_{k=1}^K\mathbf{x}_{lk}\in\mathbb{C}^{N_{BS}\times 1}, 
\end{equation*}
where $\mathbf{x}_{l,c}\in\mathbb{C}^{N_{BS}\times 1}$ is the common message, and $\mathbf{x}_{lk}\in\mathbb{C}^{N_{BS}\times 1}$ is the transmit  signal of BS $l$ intended for its $k$th associated user, i.e., u$_{lk}$. The signals $\mathbf{x}_{l,c}$ and $\mathbf{x}_{lk}$s are zero-mean uncorrelated improper Gaussian random vectors. The common message of  BS $l$ is decoded by all its associated users; however, $\mathbf{x}_{lk}$ is decoded only by u$_{lk}$. 
Moreover, each transceiver may suffer from IQI, following the model described in Appendix \ref{sec-iqi}. Employing Lemma \ref{lem:Rea} in Appendix \ref{sec-iqi}, the received signal at u$_{lk}$ is
\begin{multline}
\underline{\mathbf{y}}_{lk}=\sum_{i=1}^L\underline{\mathbf{H}}_{lk,i}
\sum_{j=1}^K\underline{\mathbf{x}}_{ij}+\underline{\mathbf{n}}_{lk}
=
\underbrace{\underline{\mathbf{H}}_{lk,l}
\underline{\mathbf{x}}_{l,c}}_{\text{Common  message}}+
\underbrace{\underline{\mathbf{H}}_{lk,l}
\underline{\mathbf{x}}_{lk}}_{\text{Private message}}
\\
+
\underbrace{\underline{\mathbf{H}}_{lk,l}
\sum_{j=1,j\neq k}^{K}\underline{\mathbf{x}}_{lj}}_{\text{Intracell interference}}
+
\underbrace{ \sum_{i=1,i\neq l}^L\underline{\mathbf{H}}_{lk,i}
\underline{\mathbf{x}}_{i}}_{\text{Intercell interference}}+
\underbrace{\underline{\mathbf{n}}_{lk}}_{\text{Noise}},
\end{multline}
where 
$\underline{\mathbf{H}}_{lk,i}
\in\mathbb{R}^{2N_u\times 2N_{BS}}$ for all $i,l,k$ is the equivalent channel given by Lemma \ref{lem:Rea}. Note that $\underline{\mathbf{x}}_{l},\underline{\mathbf{x}}_{l,c},\underline{\mathbf{x}}_{lj}\in\mathbb{R}^{2N_{BS}\times 1}$ for all $l,j$ are  $2N_{BS}\times 1$ real vectors stacking the real and imaginary parts of ${\mathbf{x}}_{l}$, ${\mathbf{x}}_{l,c}$, and ${\mathbf{x}}_{lj}$, respectively. Additionally, note that if IQI parameters are not the same at all transceivers, we can still apply Lemma \ref{lem:Rea} to easily obtain the equivalent channel for each individual links. 

We represent the transmit covariance matrix of  $\underline{\mathbf{x}}_{l}$, $\underline{\mathbf{x}}_{l,c}$, and $\underline{\mathbf{x}}_{lk}$by ${\bf P}_{l}$, ${\bf P}_{l,c}$ and ${\bf P}_{lk}$, respectively. Note that $\mathbf{P}_{l}=\mathbf{P}_{l,c}+\sum_{k}\mathbf{P}_{lk}.$ 
Moreover, we represent the feasibility set for IGS by 
\begin{equation}
\mathcal{P}_{I}\!=\!
\left\{
\mathbf{P}_{lk},\mathbf{P}_{l,c}\!:\!
\text{Tr}\left(\mathbf{P}_{l}\right)\leq p_l, \mathbf{P}_{lk}\succcurlyeq\mathbf{0}, \mathbf{P}_{l,c}\succcurlyeq\mathbf{0}, \forall l,k\right\}\!,
\end{equation}
where $p_l$ is the power budget of BS $l$. 
Note that the transmit covariance matrices can be  arbitrary positive semi-definite real matrices for IGS. However, for PGS, these matrices are structured as described in \cite[Eq. (5)]{soleymani2022improper} since the real and imaginary parts of the signal are iid in this case. Therefore, the feasibility set for PGS is 
\begin{multline}
\mathcal{P}_{P}=
\{
\mathbf{P}_{lk},\mathbf{P}_{l,c}:\text{Tr}\left(\mathbf{P}_{l}\right)\leq p_l, \mathbf{P}_{lk}\in\mathcal{P}_t,\\
\mathbf{P}_{l,c}\in\mathcal{P}_t,
\mathbf{P}_{lk}\succcurlyeq\mathbf{0}, \mathbf{P}_{l,c}\succcurlyeq\mathbf{0},\forall l,k\},
\end{multline}
where 
$\mathcal{P}_t$ is the set of all semi-definite positive real matrices that fulfill the structure in \cite[Eq. (5)]{soleymani2022improper}. 
Since our framework is general and can be applied to both PGS and IGS, we represent the feasibility set for transmit covariance matrices by $\mathcal{P}$ for notational simplicity.

\subsection{Rate and energy efficiency expressions}\label{sec-ii-c}
We assume that the intercell interference is treated as noise, and we employ the RS technique to partly manage the intracell interference. In the 1-layer RS, each user first decodes the common message while treating all other signals as noise  \cite{mao2022rate,mishra2021rate,clerckx2016rate}.  Then, each user subtracts the common message from the received signal,  and then decodes its private message treating the remaining intracell interference as well as intercell interference as noise \cite{mao2022rate,mishra2021rate,clerckx2016rate}.
\subsubsection{Rate expressions} 
The achievable rate of u$_{lk}$ is the summation of the  rate of decoding its private message and its dedicated rate from the common message, i.e.,
\begin{equation}\label{rate-eq}
r_{lk}=r_{lk,c}+r_{lk,p},
\end{equation}
 where $r_{lk,p}$ is the rate of decoding $\mathbf{x}_{lk}$ at u$_{lk}$ after decoding and canceling $\mathbf{x}_{l,c}$ and treating $\mathbf{x}_{li}$ for $i\neq k$ as noise \cite[Eqs. (2)-(3)]{mishra2021rate}
\begin{align}\label{eq-28}
r_{lk,p}=\!\frac{1}{2}\log_2\left|{\bf I}+
{\bf D}_{lk}^{-1}
 \underline{{\bf H}}_{lk,l}
{\bf P}_{lk}\underline{{\bf H}}_{lk,l}^T  
\right|\!=\!r_{lk,p_1}\!-r_{lk,p_2}\!,
\end{align}
where ${\bf P}_{lk}$ is the transmit covariance matrix of  $\underline{{\bf x}}_{lk}$, and 
\begin{align}\label{ne=10}
r_{lk,p_1}\!\!&=\!\frac{1}{2}\!\log_2\!\left|{\bf D}_{lk}
\!+\!
\underline{{\bf H}}_{lk,l}
{\bf P}_{lk}\underline{{\bf H}}_{lk,l}^T
\right|\!,
r_{lk,p_2}\!\!=\!\frac{1}{2}\!\log_2\!\left|{\bf D}_{lk}
\right|\!,\!
\\
\mathbf{D}_{lk}
&=
\underbrace{
\sum_{i=1,i \neq l}^L\underline{\mathbf{H}}_{lk,i}
\mathbf{P}_{i}\underline{\mathbf{H}}_{lk,i}^T
}_{\text{Intercell interference}}
+
\underbrace{
\sum_{j= 1,j\neq k}^{K}
\underline{\mathbf{H}}_{lk,l}
\mathbf{P}_{lj}
\underline{\mathbf{H}}_{lk,l}^T
}_{\text{Intracell interference}}
+
\underbrace{
\underline{\mathbf{C}}_{n}}_{\text{Noise}},
\end{align}
where $\underline{\mathbf{C}}_{n}$ is the noise variance given by Lemma \ref{lem:Rea}.
  Moreover, $r_{lk,c}$ is the portion of the decoding rate of common message allocated to u$_{lk}$. We denote the common rate at BS $l$ as $r_l$, which is given by 
\begin{equation}\label{fis=}
r_l=\sum_{k=1}^{K}r_{lk,c}\leq \min_{
k}\left\{\bar{r}_{lk,c}\right\}\triangleq r_{l,c},
\end{equation}
 where $\bar{r}_{lk,c}$ is the maximum decodable rate of the common message at u$_{lk}$ treating the other signals as noise \cite[Eqs. (2)-(3)]{mishra2021rate}
\begin{align}
\nonumber
\bar{r}_{lk,c}&=\frac{1}{2}\log_2\left|{\bf I}+
\left({\bf D}_{lk}+\underline{{\bf H}}_{lk,l}
{\bf P}_{lk}\underline{{\bf H}}_{lk,l}^T\right)^{-1}
 \underline{{\bf H}}_{lk,l}
{\bf P}_{l,c}\underline{{\bf H}}_{lk,l}^T  
\right|
\\ &\nonumber
=
\underbrace{
\frac{1}{2}\log_2\left|{\bf D}_{lk}+\underline{{\bf H}}_{lk,l}
{\bf P}_{lk}\underline{{\bf H}}_{lk,l}^T
+
\underline{{\bf H}}_{lk,l}
{\bf P}_{l,c}\underline{{\bf H}}_{lk,l}^T
\right|
}_{\bar{r}_{lk,c_1}
}
\\ &
\hspace{1cm}
-
\underbrace{
\frac{1}{2}\log_2\left|{\bf D}_{lk}+\underline{{\bf H}}_{lk,l}
{\bf P}_{lk}\underline{{\bf H}}_{lk,l}^T
\right|}
_{\bar{r}_{lk,c_2}
}.
\label{eq=29=}
\end{align}
Note that the common message of BS $l$ should be decodable by all its associated users. Thus, the rate has to be equal to (or less than) $\min_{
k}\left\{\bar{r}_{lk,c}\right\}$ to ensure the decodability of the common message for all users associated to the BS. It is worth emphasizing that  the rates $r_{lk,c}$ are actually optimization parameters, and it may happen that a user receives all its rate through the common message, while another user receives all its rate through the private message. We will discuss the possible solutions of $r_{lk,c}$ in the next section with more details. 

\subsubsection{Energy efficiency metrics}
The global EE (GEE) is defined as the total achievable rate (throughput) of the system, divided by the total power consumption, i.e., \cite{zappone2015energy}
\begin{equation}
GEE=\frac{\sum_{\forall l,k}r_{lk}}{LKp_c+\eta\sum_{l}\text{Tr}\left(\mathbf{P}_{l}\right)}, 
\end{equation}
where $\eta^{-1}$ is the power efficiency of each BS, and $p_c$ is the constant power consumption in the network for transmitting data to a user, which is given by \cite[Eq. (27)]{soleymani2022improper}. GEE is a metric for the performance of the whole network. Alternatively,  the EE of each user is defined as \cite{zappone2015energy}
\begin{equation}
EE_{lk}=\frac{r_{lk}}{p_c+\eta\text{Tr}\left(\mathbf{P}_{lk}\right)+\frac{\eta}{K}\text{Tr}\left(\mathbf{P}_{l,c}\right)}.
\end{equation}

\subsection{Problem statement}
In this subsection, we formulate a general optimization problem for RS in MIMO RIS-assisted systems as 
\begin{subequations}\label{ar-opt}
\begin{align}
\label{12-a}
 \underset{\{\mathbf{P}\}\in\mathcal{P},\{\bm{\Theta}\}\in\mathcal{T},\mathbf{r}_c
 }{\max}\,\,\,\,\,\,\,\,  & 
  f_0\left(\left\{\mathbf{P}\right\},\{\bm{\Theta}\}\right) &\\
 \,\,\,\,\,\,\,\,\,\,\,\, \,\, \text{s.t.}   \,\,\,\,\,\,\,\,\,\,&  f_i\left(\left\{\mathbf{P}\right\},\{\bm{\Theta}\}\right)\geq0,&\forall i,
\\
\label{4-c}
 &
r_{lk}\geq r_{lk}^{th},
&\forall l,k,
\\\label{4-b}
 &  \sum_{k=1}^Kr_{lk,c}\leq r_{l,c}(\{\mathbf{P}\}),
&\forall l,k,
\\
\label{4-d}
&
r_{lk,c}\geq 0,
&\forall l,k,
 \end{align}
\end{subequations}
where $\mathbf{r}_c=\left\{r_{lk,c},\forall l,k\right\}$ is the vector of rates for the common message, $f_i$s are linear functions of rates/EEs and/or concave/convex/linear functions of transmit covariance matrices,
  \eqref{4-c} is the rate constraint, \eqref{4-b} is the decodability constraint,  \eqref{4-d} is due to the fact that the rates cannot be negative, and
  $\mathbf{r}_c$, $\{\mathbf{P}\}$ and $\{\bm{\Theta}\}$ are the optimization parameters. 
  The function $f_0\left(\left\{\mathbf{P}\right\}\right)$ in \eqref{12-a} is the considered utility function, which can be, for instance, the [weighted] sum rate, minimum [weighted] rate, global EE and so on. 
  If we want to minimize a cost function, $f_0\left(\left\{\mathbf{P}\right\}\right)$ returns a negative value as mentioned below. 
Note that
 $f_i$  can include, for example, an energy harvesting constraint and/or the so-called interference temperature (in cognitive radio systems) since these constraints are linear/affine in $\{\mathbf{P}\}$. It is worth emphasizing  that  $r_{l,c}(\{\mathbf{P}\})=\min_{
k}\left\{\bar{r}_{lk,c}(\{\mathbf{P}\})\right\}$ is a function of $\{\mathbf{P}\}$. For notational simplicity, we drop this dependency hereafter.

The general optimization problem \eqref{ar-opt} includes the minimum-weighted-rate maximization (MWRM), weighted-sum-rate maximization (WSRM),  minimum-weighted-EE maximization (MWEEM), global EE maximization (GEEM) problems, among others. Additionally, it is possible to minimize the transmission power for a target rate  by choosing $f_0\left(\left\{\mathbf{P}\right\}\right)=-\sum_{
l}\text{Tr}\left(\mathbf{P}_l\right)$ in \eqref{ar-opt}.

\section{General framework for systems without RIS}\label{sec-cov-opt}
In this section, we propose an iterative algorithm to solve  problem \eqref{ar-opt}  for systems without RIS. 
Thus, we only optimize over the transmit covariance matrices $\{\mathbf{P}\}$ and $\mathbf{r}_c$.
In \eqref{ar-opt}, the constraints \eqref{4-b}, \eqref{4-c}, and \eqref{4-d} are linear in $\mathbf{r}_c$. 
Additionally, the rate of each user, given by \eqref{rate-eq}, is linear in $\mathbf{r}_c$. Thus, the optimization problem \eqref{ar-opt} is linear in $\mathbf{r}_c$ for fixed transmit covariance matrices.
Furthermore, the total energy/power consumption is linear in $\{\mathbf{P}\}$.
However, the optimization problem is not convex in $\{\mathbf{P}\}$ since the rates are not concave. To solve \eqref{ar-opt}, we employ MM, which is a numerical optimization approach that is briefly described in Appendix \ref{ap=MM}. To apply MM, we have to obtain suitable lower bound surrogate functions for the rates. 
To this end, we can employ the upper bound in Lemma \ref{lem-1}  and approximate the non-concave part of the rates (i.e., $-r_{lk,p_2}$ for $r_{lk,p}$, and $-\bar{r}_{lk,c_2}$ for $\bar{r}_{lk,c}$) by affine functions, which is known as convex-concave procedure (CCP) \cite{yuille2003concave}. Note that $r_{lk,p_2}$ and $\bar{r}_{lk,c_2}$ are, respectively, given by \eqref{ne=10} and \eqref{eq=29=}. For the sake of completeness, we state the concave lower bound only for $r_{lk,p}$ since it is straightforward to apply  Lemma \ref{lem-1}  to the other rate functions, i.e., $\bar{r}_{lk,c}$. 
\begin{corollary}\label{col1}
A concave lower bound only for $r_{lk,p}$ is
\begin{multline*}
\label{l-r-lk-p}
 r_{lk,p}\geq \tilde{r}_{lk,p}=
r_{lk,p_1}\left(\{\mathbf{P}\}\right) 
-r_{lk,p_2}^{(t-1)}
\\
-\sum_{j=1,\neq k}^{K}
\text{\em Tr}\left(
\frac{
\underline{\mathbf{H}}_{lk,l}^T
(\mathbf{D}_{lk}^{(t-1)})^{-1}
\underline{\mathbf{H}}_{lk,l}
}
{2\ln 2}
\left(\mathbf{P}_{lj}-\mathbf{P}_{lj}^{(t-1)}\right)
\right)
\\
-\sum_{i=1, \neq l}^L
\text{\em Tr}\left(
\frac{
\underline{\mathbf{H}}_{lk,i}^T
(\mathbf{D}_{lk}^{(t-1)})^{-1}
\underline{\mathbf{H}}_{lk,i}
}
{2\ln 2}
\left(\mathbf{P}_{i}-\mathbf{P}_{i}^{(t-1)}\right)
\right),
\end{multline*}
where $r_{lk,p_2}^{(t-1)}=r_{lk,p_2}\left(\{\mathbf{P}^{(t-1)}\}\right)$, and $\mathbf{D}_{lk}^{(t-1)}=\mathbf{D}_{lk}\left(\{\mathbf{P}^{(t-1)}\}\right)$.
\end{corollary}
\begin{proof}
$r_{lk,p_1}\left(\{\mathbf{P}\}\right)$ is concave and thus, remains unchanged. However, $-r_{lk,p_1}\left(\{\mathbf{P}\}\right)$ is convex, which is approximated by the lower bound given in Lemma \ref{lem-1}. 
To this end, we set $\mathbf{A}=\mathbf{I}$, $\underline{\mathbf{H}}_{lk,i}=\mathbf{B}_i$, which proves the corollary. 
\end{proof}
Let us represent the surrogate functions for the rates $\bar{r}_{lk,c}$ and $r_{lk}$ by, respectively, $\tilde{r}_{lk,c}$ and $\tilde{r}_{lk}$ for all $l,k$. Additionally, we define $\tilde{r}_{l,c}$ as $\tilde{r}_{l,c}=\min_{\forall k}\left\{\tilde{r}_{lk,c}\right\}$. Substituting  the surrogate functions for the rates, $\tilde{r}_{lk}$, in $f_i$ yields the surrogate functions $\tilde{f}_i$ and leads to the surrogate optimization problem
\begin{subequations}\label{ar-opt-sur}
\begin{align}
 \underset{\{\mathbf{P}\}\in\mathcal{P},\mathbf{r}_c
 }{\max}  & 
  \tilde{f}_0\left(\left\{\mathbf{P}\right\}
\right) &\\
 \text{s.t.}   \,\,\,\,\,\,\,\,\,\,&  \tilde{f}_i\left(\left\{\mathbf{P}\right\}
\right)\geq0,&\forall i.
\\
\label{14-c}
 &
\tilde{r}_{lk}=r_{lk,c}+\tilde{r}_{lk,p}(\{\mathbf{P}\})\geq r_{lk}^{th},
&\forall l,k,
\\\label{14-b}
 &  \sum_{k=1}^Kr_{lk,c}\leq \tilde{r}_{l,c}=\min_{\forall k}\left\{\tilde{r}_{lk,c}(\{\mathbf{P}\})\right\},
&\forall l,k,
\\
\label{14-d}
&
r_{lk,c}\geq 0,
&\forall l,k.
 \end{align}
\end{subequations}
Note that $f_i$ can be a linear function of the rates/EEs and/or transmit covariance matrices.
The optimization problem \eqref{ar-opt-sur} is jointly convex in $\{\mathbf{P}\}$ and $\mathbf{r}_c$ for spectral efficiency maximization problems such as the WMRM, WSRM as well as minimizing power for a target rate. Therefore, it  can be efficiently solved by existing numerical tools. For the optimization problems with EE metrics such as MWEEM and GEEM, the optimization problem \eqref{ar-opt-sur} falls into fractional programming problems \cite{zappone2015energy}, and  their global optimal solutions  can be obtained by Dinkelbach-based algorithms since each fractional function has a concave numerator and a linear denominator \cite{zappone2015energy}. For the sake of completeness, we shortly discuss the solution of each problem in the following subsections. Note that this optimization framework converges to a stationary point of \eqref{ar-opt} since the proposed algorithm falls into MM and fulfills the conditions in Lemma \ref{lem-sur-jadid}.

\subsection{Spectral efficiency maximization}
The most general performance metric for spectral efficiency is the achievable rate region, which gives all the possible operational points for a scheme. The achievable rate region is defined as the set of all achievable rates by the considered scheme:
\begin{equation*}
\mathcal{R}= \underset{\{\mathbf{P}\}\in\mathcal{P},\{\bm{\Theta}\}\in\mathcal{T},\mathbf{r}_c\in\Omega}{\bigcup}\left\{r_{lk}\right\}_{\forall l,k}, 
\end{equation*}
where $\Omega$ is the feasibility set for $\mathbf{r}_c$ 
\begin{equation}
\Omega=\left\{r_{lk,c}: r_{lk,c}\geq 0, \sum_{k=1}^K r_{lk,c}\leq r_{l,c},\forall l,k\right\}.
\end{equation}
Employing the rate-profile technique, we can obtain the achievable rate region by solving  the MWRM problem \cite{zeng2013transmit} 
\begin{align}\label{ar-opt-mwrm}
 \underset{\{\mathbf{P}\}\in\mathcal{P},\mathbf{r}_c\in\Omega,r
 }{\max}  & 
r&
&
 \text{s.t.}   &  r_{lk}=r_{lk,c}+r_{lk,p}\geq \lambda_{lk}r, &\,\,\forall l,k,
 \end{align}
and varying $\lambda_{lk}$s for all possible $\lambda_{lk}\geq 0$ and $\sum_{\forall l,k}\lambda_{lk}=1$. 
Substituting the rates by $\tilde{r}_{lk}$,  the following  convex problem  results
\begin{subequations}\label{ar-opt-mwrm-sur}
\begin{align}
 \underset{\{\mathbf{P}\}\in\mathcal{P},\mathbf{r}_c
,r
 }{\max}\,   
  r \,\,\,\,\,\,\, \text{s.t.}   \,\,\,\,\,&\tilde{r}_{lk}=r_{lk,c}+\tilde{r}_{lk,p}\geq r,&\forall l,k,\\
  &  
\eqref{14-b},\eqref{14-d},
 \end{align}
\end{subequations}
which can be efficiently solved. 

There are other common performance metrics for the spectral efficiency such as the weighted-sum rate, and/or the geometric mean of rates (GMR). The WSRM problem has a structure similar to \eqref{ar-opt-mwrm}. Hence, it is straightforward to solve the WSRM problem by our framework. Furthermore, it is shown in \cite[Eq. (10)]{yu2021maximizing} that maximizing the GMR is equivalent to solving a sequence of WSRM problems. Therefore, our framework can be easily applied to the GMR maximization problem as well. 

\subsection{Energy efficiency maximization}
In this subsection, we provide the solution for the GEEM problem and refer the readers to \cite{zappone2015energy, soleymani2020improper, soleymani2022improper} for solving MWEEM. 
The GEEM problem is
\begin{align}\label{opt--gee}
\underset{\{\mathbf{P}\}\in\mathcal{P},\mathbf{r}_c\in\Omega
 }{\max}\, & 
  \frac{\sum_{l,k}r_{lk}}{LKp_c+\eta\sum_{l}\text{Tr}\left(\mathbf{P}_{l}\right)}\!\!\!
&\text{s.t.}\,
r_{lk}\geq r_{lk}^{th},&\,\forall l,k.
\end{align}
 Thus, the surrogate optimization problem for the GEEM problem is 
\begin{align}\label{ar-opt-sur-gee}
 \underset{\{\mathbf{P}\}\in\mathcal{P},\mathbf{r}_c
 }{\max}\,\,  & 
  \frac{\sum_{
l,k}\tilde{r}_{lk}}{LKp_c+\eta\sum_{l}\text{Tr}\left(\mathbf{P}_{l}\right)} &
 \!\!\! \text{s.t.}   \,&  \eqref{14-c},\eqref{14-b},\eqref{14-d}.
 \end{align}
The global optimal solution of \eqref{ar-opt-sur-gee} can be obtained by iteratively solving \cite{zappone2015energy}
\begin{subequations}\label{ar-opt-sur-gee-2}
\begin{align}
 \underset{\{\mathbf{P}\}\in\mathcal{P},\mathbf{r}_c
 }{\max}  & 
  \sum_{
l,k}\tilde{r}_{lk}\!-\!\mu^{(m)}\!\!\left(\!LKp_c+\eta\sum_{l}\text{Tr}\left(\mathbf{P}_{l}\right)\!\right),\\
  \,\,\,\,\,\,\,\,\,\,\,\, \,\, \text{s.t.}   \,\,\,\,\,\,\,\,\,\,&  \eqref{14-c},\eqref{14-b},\eqref{14-d},
 \end{align}
\end{subequations}
and updating $\mu^{(m)}$ as
\begin{equation*}
\mu^{(m)}=\frac{\sum_{
l,k}\tilde{r}_{lk}^{(m-1)}}{LKp_c+\eta\sum_{l}\text{Tr}\left(\mathbf{P}_{l}^{(m-1)}\right)}, 
\end{equation*}
where $\mathbf{P}_{lk}^{(m-1)}$ is the solution of \eqref{ar-opt-sur-gee-2} at $(m-1)$th iteration, and $\tilde{r}_{lk}^{(m-1)}$ is the value of $\tilde{r}_{lk}$ at the end of $(m-1)$th iteration. Note that this procedure is known as the Dinkelbach algorithm \cite{zappone2015energy}. It is worth indicating that $r_{lk}^{th}$s for all $l,k$ have to be chosen such that \eqref{opt--gee} is feasible.

\subsection{Total transmission power minimization}
The total transmission power minimization for a target rate can be formulated as 
\begin{align}\label{ar-opt-mttp}
 \underset{\{\mathbf{P}\}\in\mathcal{P},\mathbf{r}_c\in\Omega
 }{\min}  & 
\sum_{l}\text{Tr}\left(\mathbf{P}_{l}\right)&
 \text{s.t.}   \,\,\,&  r_{lk}\geq r_{lk}^{th}
  &\forall l,k,
 \end{align}
which is non-convex due to the structure of the rate functions.
Replacing $r_{lk}$ with $\tilde{r}_{lk}$ and $r_{lk,c}$ with $\tilde{r}_{lk,c}$, we have
\begin{align}\label{ar-opt-mttp-sur}
 \underset{\{\mathbf{P}\}\in\mathcal{P},\mathbf{r}_c
 }{\min}\,\,\,\,\,\,\,\,  & 
\sum_{
l}\text{Tr}\left(\mathbf{P}_{lk}\right)&
 \text{s.t.}   \,\,\,\,\,\,\,\,\,\,&  \eqref{14-c},\eqref{14-b},\eqref{14-d}.
 \end{align}
The problem \eqref{ar-opt-mttp-sur} is convex. Note that the target rates should be chosen such that the optimization problem \eqref{ar-opt-mttp} is feasible. Indeed, the target rate should be achievable by the employed scheme. 

\subsection{Discussions on the solution for $\mathbf{r}_c$}
In this subsection, we shortly describe the solutions by RS and mention how RS can be specialized to TIN, NOMA, multicasting and/or broadcasting. For simplicity, we consider the 1-layer RS for a single-cell SISO BC with only two users. Let us denote $r_{c}$ the common rate, $r_{c,i}$ the portion of the common rate dedicated to user $i$, $r_{p,i}$ the rate of the private message for user $i$. Based on the value of $r_{c}$, $r_{c,i}$, and $r_{p,i}$, we can have either TIN, NOMA or broadcasting as  follows:
\begin{itemize}
\item {\bf TIN}: If $r_{c}=0$, then the RS scheme is equivalent to TIN.
 
\item {\bf NOMA}: If $r_{c,1}=r_{p,2}=0$, RS is equivalent to NOMA. The reason is that, in this case, user 1 firstly decodes the signal of user 2, and then subtracts it from the received signal. However, user 2 treats the private message of user 1 as noise and decodes only the common message. Moreover, there is no private message for user 2 to be decoded, which is the same as in NOMA. 

\item {\bf Broadcasting/Multi-casting}: If $r_{p,1}=r_{p,2}=0$, there is no private message, and both users decode only the common message, which is equivalent to  broadcasting. 
\end{itemize}
Note that for $K>2$, 1-layer RS may not necessarily attain the optimal NOMA in SISO systems since NOMA involves multiple decoding and canceling signals (multi-layer SIC) while 1-layer RS has only 1 layer SIC at each user. Therefore, when $K>2$ in SISO systems, the generalized RS, which contains multiple SIC levels, is expected to outperform 1-layer RS \cite{mao2018rate}. 

\subsection{Discussions on computational complexity}
In this subsection, we discuss the computational complexity of our framework and provide an approximate upper bound for the number of required multiplications to obtain a solution for our framework.  
The exact complexity of our proposed framework depends the number of required iterations for the convergence and the computational complexities of each iteration, which highly depends on the  implementation of the framework. 
We set the maximum number of iterations of our framework to $N$. Now we discuss the complexity of solving the surrogate optimization problem in each iteration, which is provided in \eqref{ar-opt-sur}. 
Since \eqref{ar-opt-sur} includes the most general optimization problem, we consider the MWRM problem in \eqref{ar-opt-mwrm-sur}, which has a specific structure. Note that it is straightforward to extend to the computational complexity analysis to other optimization problems. 

The surrogate optimization problem \eqref{ar-opt-mwrm-sur} is convex. The number of Newton iterations to solve a convex optimization problem is proportional to the square root of the number of its constraints \cite[Chapter 11]{boyd2004convex}.  The problem \eqref{ar-opt-mwrm-sur} has $3LK$ constraints. In each Newton iteration, the surrogate functions for the private and common messages at each user should be computed, which are in total $2LK$ surrogate functions. As shown in Corollary \ref{col1}, each surrogate function includes a logarithmic and a linear part. 
Now we obtain an approximation for the number of multiplications to compute each surrogate function. 
The computational complexity of obtaining the determinant of $n\times n$ matrix is $\mathcal{O}(n^3)$. Additionally, the number of multiplication in computing the multiplication of the matrices $A_{n_1 \times n_2}$ and $B_{n_2 \times n_3}$ is $n_1n_2n_3$. Thus, the computational complexity of obtaining $\tilde{r}_{lk,c}$ can be approximated as $\mathcal{O}(N_u^3+LN^2_{BS}\left(N_u+N_{BS}\right))$. Note that the number of multiplications for computing the rate of private messages is approximately in the same order of computing the rate of common messages. Thus, the total computational complexity of solving the WMRM problem with our framework is approximately
$\mathcal{O}\left(NLK\sqrt{LK}\left(N_u^3+LN^2_{BS}\left(N_u+N_{BS}\right)\right)\right).$

\section{Extension to RIS-assisted systems}\label{sec=opt-ris}
In this section, we consider RIS-assisted systems in which we have to optimize over transmit covariance matrices and reflecting coefficients. 
Indeed, in addition to optimizing the transmission parameters, we can optimize the environment through RISs, which is mainly the new part of the proposed framework in this section comparing to the framework in Section \ref{sec-cov-opt}. Jointly optimizing the  transmit covariance matrices and reflecting coefficients entails many challenges. For instance, finding suitable surrogate functions is more complicated when optimizing over RIS components. Additionally, the  feasibility sets for RIS components are mainly non-convex, which further complicates the optimization problems. 
To tackle the challenges, we employ an AO approach to separate the optimization of transmit covariance matrices from RIS components. That is, we first optimize over transmit covariance matrices for a fixed $\{\bm{\Theta}^{(t-1)}\}$ and obtain the new $\{\mathbf{P}^{(t)}\}$. Then, we update  $\{\bm{\Theta}\}$ by fixing the transmit covariance matrices to $\{\mathbf{P}^{(t)}\}$. In the following, we describe each optimization steps in separate subsections. 

\subsection{Optimizing the transmit covariance matrices}
For fixed reflection coefficients $\{\bm{\Theta}^{(t-1)}\}$, the channels are fixed, and the optimization problems are equivalent to the problems in Section \ref{sec-cov-opt}.
Thus, we can employ the framework in Section \ref{sec-cov-opt} to optimize over the transmit covariance matrices. Due to space limitations, we do not repeat them in this subsection again and refer the reader to Section \ref{sec-cov-opt}.

\subsection{Optimizing the reflecting coefficients}
For given transmit covariance matrices $\{\mathbf{P}^{(t)}\}$, the optimization problem \eqref{ar-opt} can be written as 
\begin{subequations}\label{ar-opt-2-rf}
\begin{align}
 \underset{\{\bm{\Theta}\}\in\mathcal{T},\mathbf{r}_c
 }{\max}\,\,\,\,\,\,\,\,  & 
  f\left(
\{\bm{\Theta}\}\right) 
&
  \text{s.t.}   \,\,\,\,\,\,\,\,\,\,& f_i\left(
\{\bm{\Theta}\}\right) &\forall i,
\\
\label{4-c=}
 &&&
r_{lk}\geq r_{lk}^{th},&\forall l,k,
\\
\label{4-b=}
&&& \sum_k
r_{lk,c}\leq r_{l,c},&\forall l,\\
\label{4-d=}
&&&
r_{lk,c}\geq 0,&\forall l,k,
 \end{align}
\end{subequations}
where the constraints  \eqref{4-c=}, \eqref{4-b=}, and \eqref{4-d=} are linear in $\mathbf{r}_c$. However, the corresponding optimization problem is not convex due to the structure of the rate functions and the structure of the feasibility sets. To solve \eqref{ar-opt-2-rf}, we first obtain suitable lower-bound surrogate functions for the rates by using the inequality in Lemma \ref{lem-2}. In the following, we only write the concave lower bound for  $r_{lk,p}$ since it is straightforward to apply Lemma \ref{lem-2} to derive concave lower bounds for $r_{lk,c}$.
\begin{corollary}\label{cor2}
A concave lower bound for $r_{lk,p}$ is
\begin{multline*} 
r_{lk,p}\geq
 \hat{r}_{lk,p}=r_{lk,p}^{(t-1)}-
\frac{1}{2\ln 2}\text{\em{ Tr}}\left(
\bar{\mathbf{V}}_{lk,p}\bar{\mathbf{V}}_{lk,p}^T\bar{\mathbf{Y}}_{lk,p}^{-1}
\right)
\\
+
\frac{1}{\ln 2}
\text{\em{ Tr}}\left(
\bar{\mathbf{V}}_{lk,p}^T\bar{\mathbf{Y}}_{lk,p}^{-1}\mathbf{V}_{lk,p}
\right)
-
\frac{1}{2\ln 2}\times
\\
\text{\em{ Tr}}\left(\!
(\bar{\mathbf{Y}}^{-1}_{lk,p}\!-(\bar{\mathbf{V}}_{lk,p}\bar{\mathbf{V}}^T_{lk,p}\! +\! \bar{\mathbf{Y}}_{lk,p})^{-1})^T
(\mathbf{V}_{lk,p}\mathbf{V}^T_{lk,p}\!+\!\mathbf{Y}_{lk,p})\!
\right)\!,
\end{multline*}
where $\mathbf{Y}_{lk,p}=\mathbf{D}_{lk}\left(
\{\bm{\Theta}\}\right) $, $\mathbf{V}_{lk,p}=\underline{\mathbf{H}}_{lk,l}\left(\{\bm{\Theta}\}\right)\mathbf{P}_{lk}^{(t)^{1/2}}$, and
\\
\begin{align*}\bar{\mathbf{V}}_{lk,p}&=\underline{\mathbf{H}}_{lk,l}\left(\{\bm{\Theta}\}^{(t-1)}\right)\mathbf{P}_{lk}^{(t)^{1/2}}\!,\!\!  
&
\bar{\mathbf{Y}}_{lk,p}&=\mathbf{D}_{lk}\left(\!\!
\{\bm{\Theta}\}^{(t-1)}\right)\!.
\end{align*}
\end{corollary}
Let us denote the concave lower bounds for $r_{lk,c}$ as $\hat{r}_{lk,c}$. The surrogate optimization problem for weighted-sum and minimum-weighted rates, minimum-weighted EE and global EE maximization  problems can be formulated as
\begin{subequations}\label{ar-opt-2-rf-2}
\begin{align}
 \underset{\{\bm{\Theta}\}\in\mathcal{T},\mathbf{r}_c
 }{\max}\,\,\,\,\,\,\,\,  & 
  \hat{f}\left(\left\{\mathbf{P}^{(t)}\right\},\{\bm{\Theta}\}\right) &\\
\label{4-b=+}
 \,\,\,\,\,\,\,\,\,\,\,\, \,\, \text{s.t.}   \,\,\,\,\,\,\,\,\,\,&  \sum_{k=1}^Kr_{lk,c}\leq \hat{r}_{l,c}=\min_{
k}\left\{\hat{r}_{lk,c}\right\},&\forall l,\\
\label{4-c=+}
 &
\hat{r}_{lk}=r_{lk,c}+\hat{r}_{lk,p}\geq r_{lk}^{th},&\forall l,k,
\\
\label{4-d=+}
&
r_{lk,c}\geq 0,&\forall l,k,
 \end{align}
\end{subequations}
where the utility function is a linear function of the rates $\hat{r}_{lk}$. The reason is that the EE of a user is a scale of the rate of the user for fixed transmit covariances. 
Note that we cannot minimize the total transmit power for a target rate when transmit covariance matrices are fixed. Thus, for the total-transmission-power-minimization problem, we maximize the minimum achievable rate of users since it allows to further decrease the transmission power in the next iteration. 

When the feasibility set for the reflecting coefficients is convex, the optimization problem \eqref{ar-opt-2-rf-2} is convex as well. Since $\mathcal{T}_U$ is a convex set, \eqref{ar-opt-2-rf-2} is convex for $\mathcal{T}_U$ and can be solved efficiently. However, we have to convexify the sets  $\mathcal{T}_I$ and $\mathcal{T}_C$.
In the following, we discuss how we can deal with the feasibility set for reflecting coefficients.
Note that whole algorithm for the feasibility set $\mathcal{T}_U$ converges to a stationary point of \eqref{ar-opt}. 

 \subsubsection{Feasibility set $\mathcal{T}_I$} In this feasibility set, we have the constraint $|\theta_{mn}|=1$ for all $m,n$, which can be written as the following two constraints
$|\theta_{mn}|^2\leq 1$, and 
$|\theta_{mn}|^2\geq 1$. 
The constraint $|\theta_{mn}|^2\leq 1$ is convex. However, the constraint $|\theta_{mn}|^2\geq 1$ is not convex since the function $|\theta_{mn}|^2$ is convex rather than concave. Thus, we apply Lemma \ref{lem=5} to convexify this constraint as
\begin{equation*}
|\theta_{mn}|^2\geq|\theta_{mn}^{(t-1)}|^2-2\mathfrak{R}\{\theta_{mn}^{(t-1)^*}(\theta_{mn}-\theta_{mn}^{(t-1)})\}\geq 1,
\end{equation*}
where $\theta_{mn}^{(t-1)}$ is the value of $\theta_{mn}$ at $(t-1)$-th iteration. 
We then relax this constraint to make the convergence faster as
\begin{equation}\label{+=+}
|\theta_{mn}|^2\!\geq\!|\theta_{mn}^{(t-1)}|^2\!-2\mathfrak{R}\{\theta_{mn}^{(t-1)^*}(\theta_{mn}-\theta_{mn}^{(t-1)})\}\!\geq 1-\epsilon,
\end{equation}
for all $m,n$, where $\epsilon>0$. Finally, the convex surrogate optimization problem for this feasibility set is
 \begin{subequations}\label{ar-opt-2-rf-2-i}
\begin{align}
 \underset{\{\bm{\Theta}\},\mathbf{r}_c
 }{\max}  & 
  \hat{f}\left(\left\{\mathbf{P}^{(t)}\right\},\{\bm{\Theta}\}\right) &
    \text{s.t.}  \,&  \eqref{4-b=+}\!-\!\eqref{4-d=+}, \eqref{+=+},
    \\
    &&&|\theta_{mn}|^2\leq 1\,\,\forall m,n,
 \end{align}
\end{subequations}
which can be efficiently solved. Since we relaxed the constraint $|\theta_{mn}|^2\geq 1$, the solution given by \eqref{ar-opt-2-rf-2-i}, i.e.,   $\{\hat{\bm{\Theta}}\}$ is not necessarily feasible. Thus, we project  $\{\hat{\bm{\Theta}}\}$ to the feasibility set $\mathcal{T}_I$ by normalizing $\{\hat{\bm{\Theta}}\}$ as $\{\hat{\bm{\Theta}}^{\text{new}}\}$, i.e.,
\begin{equation}\label{eq==27}
\theta_{mn}^{\text{new}}=\frac{\hat{\theta}_{mn}}{|\hat{\theta}_{mn}|}, \hspace{1cm}\forall m,n.
\end{equation}
 To ensure the convergence of our scheme, we update the reflecting coefficients as
\begin{equation}\label{eq-42}
\{\bm{\Theta}^{(t)}\}=
\left\{
\begin{array}{lcl}
\{\hat{\bm{\Theta}}^{\text{new}}\}&\text{if}&
f\left(\left\{\mathbf{P}^{(t)}\right\},\{\hat{\bm{\Theta}}^{\text{new}}\}\right)\geq
\\
&&
f\left(\left\{\mathbf{P}^{(t)}\right\},\{\bm{\Theta}^{(t-1)}\}\right)
\\
\{\bm{\Theta}^{(t-1)}\}&&\text{Otherwise}.
\end{array}
\right.
\end{equation}
By this updating rule, the algorithm generates a sequence of non-decreasing $f$, which guarantees  convergence. 

\subsubsection{Feasibility set $\mathcal{T}_{C}$} We propose a suboptimal scheme for this feasibility set by convexifying  $\mathcal{T}_{C}$. To this end, we first relax the relationship between the phase and amplitude of reflecting components, which gives us the following two constraints 
\begin{align}\label{eq=9-2}
 |\theta_{mn}|^2&\leq 1,\\
 |\theta_{mn}|^2&\geq|\theta|_{\min}^2,
\label{eq=9-20}
\end{align}
for all $m,n$.
\eqref{eq=9-2} is convex as indicated, but \eqref{eq=9-20} is not a convex constraint since  $|\theta_{mn}|^2$ is convex, instead of concave.
Thus, we find a surrogate function for $|\theta_{mn}|^2$  by Lemma \ref{lem=5} as
\begin{equation}\label{eq-50-2}
|\theta_{mn}^{(t-1)}|^2+2\mathfrak{R}\left(\theta_{mn}^{(t-1)}(\theta_{mn}-\theta_{mn}^{(t-1)})^*\right)\geq |\theta|_{\min}^2.
\end{equation}
 Substituting these constraints in \eqref{ar-opt-2-rf-2}, we have
\begin{align}\label{ar-opt-2-rf-2-i=}
 \underset{\{\bm{\Theta}\},\mathbf{r}_c
 }{\max}  & 
  \hat{f}\left(\left\{\mathbf{P}^{(t)}\right\},\{\bm{\Theta}\}\right) &
    \text{s.t.}  \,&  \eqref{4-b=+}\!-\!\eqref{4-d=+}, \eqref{eq=9-2},\eqref{eq-50-2},
 \end{align}
which is a convex problem. 
The solution of \eqref{ar-opt-2-rf-2-i=} $\{\bm{\Theta}^{\star}\}$ is not necessarily feasible. Thus we project $\{\bm{\Theta}^{\star}\}$ into $\mathcal{T}_{C}$ by choosing 
\begin{equation}
\{\hat{\bm{\Theta}}^{\text{new}}\}=\mathcal{F}(\angle\{\bm{\Theta}^{\star}\}),
\end{equation}
where $\mathcal{F}$ is defined as in \eqref{eq*=*}.
We update $\{\bm{\Theta}\}$ according to \eqref{eq-42} to ensure the convergence of the scheme. 

\subsubsection{Feasibility set $\mathcal{T}_{D}$} We propose a suboptimal algorithm for the feasibility set $\mathcal{T}_{D}$. That is, we first relax the  and assume that the phases of the RIS components are continuous and can take any value similar to the feasibility set $\mathcal{T}_{I}$. Then, we solve  \eqref{ar-opt-2-rf-2-i} and  obtain $\{\hat{\bm{\Theta}}\}$. 
Now we project $\{\hat{\bm{\Theta}}\}$ to the feasibility set $\mathcal{T}_{D}$. To this end, we first normalize $\{\hat{\bm{\Theta}}\}$ according to \eqref{eq==27} and then take the closest phase to $\angle\hat{\theta}_{mn}$ in $\mathcal{T}_{D}$ to obtain $\{\hat{\bm{\Theta}}^{\text{new}}\}$. Finally, we update $\{\bm{\Theta}^{(t)}\}$ according to the rule in \eqref{eq-42} to ensure the convergence of our algorithm. 
\subsubsection{Extension to STAR-RIS}
The structure of the general problem for STAR-RIS is very similar to the case with conventional RIS. Indeed, the only difference is that some reflecting coefficients are correlated to each other according to  \eqref{model-1} or \eqref{model-2}. As indicated, the constraint \eqref{model-1} is convex. Thus, the following optimization problem is convex
\begin{align*}
 \underset{\{\bm{\Theta}\},\mathbf{r}_c
 }{\max}  & 
  \hat{f}\left(\left\{\mathbf{P}^{(t)}\right\},\{\bm{\Theta}\}\right) &
    \text{s.t.}  \,&  \eqref{4-b=+}\!-\!\eqref{4-d=+}, 
    \\
    &&&|\theta_{mn}^t|^2+|\theta_{mn}^r|^2\leq 1\,\,\forall m,n,
 \end{align*}
which can be solved efficiently. As a result, our proposed framework converges to a stationary point of \eqref{ar-opt} for the STAR-RIS-assisted MIMO systems with the convex constraint $|\theta_{mn}^t|^2+|\theta_{mn}^r|^2\leq 1$.

The constraint $|\theta_{mn}^t|^2+|\theta_{mn}^r|^2= 1$ is not convex, but we can convexify it similar to our approach for the feasibility set $\mathcal{T}_{I}$. That is, we rewrite the constraint as the two following constraints
\begin{align}\label{eq33}
|\theta_{mn}^t|^2+|\theta_{mn}^r|^2&\leq 1,\\
|\theta_{mn}^t|^2+|\theta_{mn}^r|^2&\geq 1,
\label{eq34}
\end{align}
for all $m,n$.
The constraint \eqref{eq33} is convex as indicated. However, \eqref{eq34} is not a convex constraint since $|\theta_{mn}^t|^2$ and $|\theta_{mn}^r|^2$ are convex functions, instead of being concave. Thus, we employ Lemma \ref{lem=5} to convexify \eqref{eq34} as 
\begin{multline}
 |\theta_{mn}^t|^2+|\theta_{mn}^r|^2\geq g^{\text{S}}_{mn}\triangleq|\theta_{mn}^{t^{(t-1)}}|^2+|\theta_{mn}^{r^{(t-1)}}|^2
 \\
 +2\mathfrak{R}\left(\theta_{mn}^{t^{(t-1)}}(\theta_{mn}^t-\theta_{mn}^{t^{(t-1)}})^*
+
\theta_{mn}^{t^{(t-1)}}(\theta_{mn}^t-\theta_{mn}^{t^{(t-1)}})^*
\right)
\\
\geq 1,
\end{multline}
where $\theta_{mn}^{t^{(t-1)}}$/$\theta_{mn}^{r^{(t-1)}}$ is the value of $\theta_{mn}^{t}$/$\theta_{mn}^{r}$ at the previous step. 
To make the convergence faster, we relax it by rewriting it as 
\begin{equation}\label{cnt=+-}
g^{\text{S}}_{mn}\geq 1-\epsilon.
\end{equation}
Plugging the constraints \eqref{eq33} and \eqref{cnt=+-} into \eqref{ar-opt-2-rf-2} results in the following convex problem
\begin{align}\label{ar-opt-2-rf-2-i=star}
 \underset{\{\bm{\Theta}\},\mathbf{r}_c
 }{\max}  & 
  \hat{f}\left(\left\{\mathbf{P}^{(t)}\right\},\{\bm{\Theta}\}\right) &
    \text{s.t.}  \,&  \eqref{4-b=+}\!-\!\eqref{4-d=+}, \eqref{eq33},\eqref{cnt=+-},
 \end{align}
whose solution $\{\hat{\bm{\Theta}}^r\}$, $\{\hat{\bm{\Theta}}^t\}$ may not satisfy \eqref{model-1}. Thus, we normalize $\{\hat{\bm{\Theta}}^r\}$, $\{\hat{\bm{\Theta}}^t\}$ as in \eqref{eq==27} to obtain  $\{\hat{\bm{\Theta}}^r_{\text{new}}\}$, $\{\hat{\bm{\Theta}}^t_{\text{new}}\}$. To ensure the convergence, we update the reflecting coefficients as
in \eqref{eq-42}.
\subsection{Discussion on computational complexity analysis}
Our proposed framework for RIS-assisted systems is iterative, and each iteration consists of two steps.
The first step, is to optimize the transmit covariance matrices, which has the same computational complexity as one iteration of our proposed framework in Section \ref{sec-cov-opt}. For instance, the computational complexity of updating the transmit covariance matrices for the WMRM problem can be approximated as $\mathcal{O}\left(LK\sqrt{LK}\left(N_u^3+LN^2_{BS}\left(N_u+N_{BS}\right)\right)\right)$. To optimize RIS components, we have to solve the surrogate optimization problem \eqref{ar-opt-2-rf-2}.
In this paper, we have considered three different feasibility sets for RIS components. Additionally, we consider the STAR-RIS. 
In this subsection, we provide an approximation for optimizing the RIS components in STAR-RIS-assisted systems. Note that the algorithms for the other feasibility sets have approximately the same computation complexity. Due to a strict space restriction, we skip the computation complexity analysis for the algorithms with different feasibility sets since they can be similarly obtained.  

To update the RIS components for STAR-RIS-assisted systems, we have to solve the convex optimization problem \eqref{ar-opt-2-rf-2-i=star}, which has at least $2LK+L+2N_{RIS}M$ constraints. 
The term $L$ can be neglected, comparing to $2LK+2N_{RIS}M$. Thus, the number of Newton iterations to solve \eqref{ar-opt-2-rf-2-i=star} approximately grows with $\sqrt{LK+N_{RIS}M}$. The main complexity in each Newton iteration is to obtain the concave lower bounds for the rate expressions in Corollary \ref{cor2}, which are quadratic in channels. In other words, the number of multiplications in computing all the surrogate rate functions is considerably higher than the number of other multiplications in order to solve \eqref{ar-opt-2-rf-2-i=star}. 
Note that the constant coefficients in the lower bound in Corollary \ref{cor2} can be computed once. Hence, the main complexities are the multiplications that are repeated in each Newton iteration. 
As a result, the complexity of solving \eqref{ar-opt-2-rf-2-i=star} can be approximated as:
\begin{multline}\mathcal{O}\left(LK\sqrt{LK+N_{RIS}M}\left(LN_u\left(N_u^2+N_{BS}^2\right)\right.\right.
\\
\left.\left.+N_uN_{BS}\left(N_u+N_{BS}\right)+MN_{RIS}N_uN_{BS}\right)\right).
\end{multline}
Finally, the computational complexity of our framework to solve the MWRM problem for STAR-RIS-assisted systems can be approximated as
 \begin{multline}\mathcal{O}\left(NLK\left[\sqrt{LK+N_{RIS}M}\left(LN_u\left(N_u^2+N_{BS}^2\right)
\right.\right.\right.
\\
\left.+N_uN_{BS}\left(N_u+N_{BS}\right)
+MN_{RIS}N_uN_{BS}\right)
\\
\left.\left.+\sqrt{LK}\left(N_u^3+LN^2_{BS}\left(N_u+N_{BS}\right)\right)\right]\right).
\end{multline}

\section{Numerical results}\label{sec-num}
\begin{figure}[t!]
    \centering
\includegraphics[width=.5\textwidth]{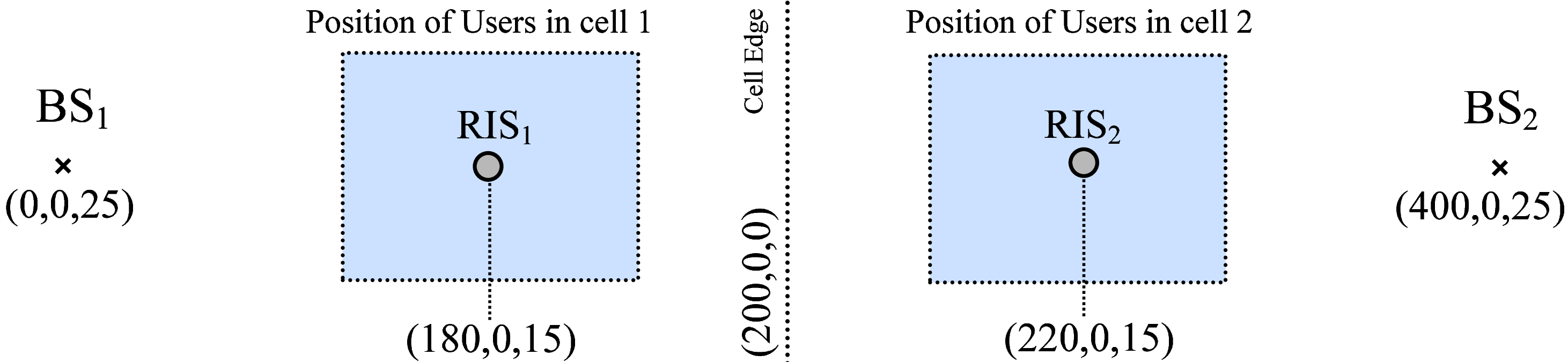}
     \caption{System topology in simulations.}
	\label{Fig-sim-model}
\end{figure}

In this section, we provide some numerical examples for WMRM, MWEEM, and for minimizing the total transmission power for a target rate. In this work, we focus on the performance of RSMA techniques in different operational regimes, analyzing the effect of different system parameters such as the number of BS/user antennas and the number of users per cell. 

In the simulations, we consider a two-cell system with one RIS in each cell. The BSs have $25$ meters height and are located at $(0,0,25)$ and $(400,0,25)$. Furthermore, RISs have $15$ meters height and are located at $(180,0,15)$ and $(220,0,15)$. 
We assume that there are $K$ users at a height of $1.5$ meters  in each cell, which are located in a square with a RIS at its center as depicted in Fig. \ref{Fig-sim-model}. The sides of the squares are $20$ meters long. The power budget for all BSs is equal to $P$. For the feasibility set $\mathcal{T}_{D}$, we assume that there are $4$ phase shifters, which means that there can be $16$ possible discrete phases. 
Other simulation  parameters are chosen based on \cite{soleymani2022improper}. 

To the best of our knowledge, this is the first work on RS in RIS-assisted MIMO systems with IQI. Hence, we compare our proposed algorithms with the IGS/PGS schemes with TIN in \cite{soleymani2022improper} as well as the suboptimal NOMA-based IGS scheme in \cite{soleymani2022noma}. Note that there is no other work on IGS in combination with NOMA in MIMO systems. Thus, we compare our algorithms with the suboptimal NOMA-based IGS scheme in \cite{soleymani2022noma}, which was proposed for multicell MISO RIS-assisted BCs. To summarize, in the simulations we consider the following transmission schemes:
\begin{itemize}
\item 
{\bf IR$_X$R} (or {\bf PR$_X$R}) refers to the IGS (or PGS) scheme with RS for $\mathcal{T}_X$, where X can be equal to U, I and C. Also, we represent
 the IGS (or PGS) scheme with RS and random reflecting coefficients by ``IR$_R$R'' (or  PR$_R$R).
\item 
{\bf IR$_U$T} (or {\bf PR$_U$T}) refers to the proposed IGS (or PGS) scheme in \cite{soleymani2022improper} for the feasibility set $\mathcal{T}_U$ with TIN. Note that this scheme can be seen as an upper bound for the IGS (or PGS) performance with TIN.
\item 
{\bf IN} (or {\bf PN}) refers to the proposed IGS (or PGS) scheme with RS, but without RIS.
\item 
{\bf PR$_{I}$R$_{\text{un}}$} refers to the PGS scheme, which is unaware of IQI, with RIS and the feasibility set $\mathcal{T}_I$.
\item 
{\bf ISRR} refers to the IGS scheme with RS for STAR-RIS assisted systems with the feasibility constraint in \cite{wu2021coverage}. 

\end{itemize}
\begin{figure}[t!]
    \centering
    \begin{subfigure}[t]{0.5\textwidth}
        \centering
\includegraphics[width=.9\textwidth]{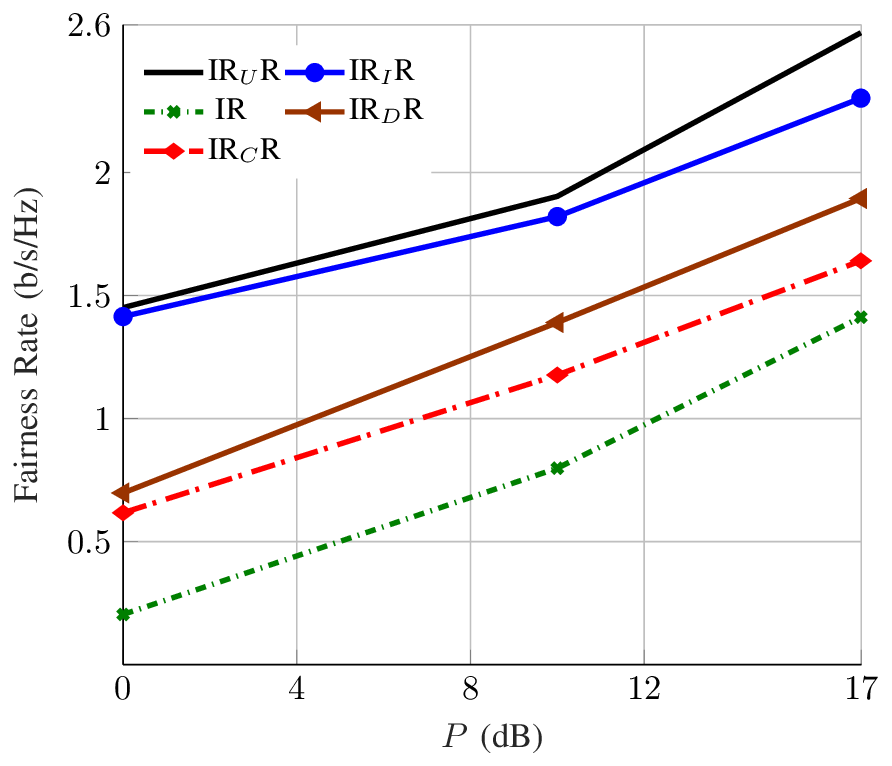}
        \caption{IGS schemes.}
    \end{subfigure}%
    \\ 
    \begin{subfigure}[t]{0.5\textwidth}
        \centering
\includegraphics[width=.9\textwidth]{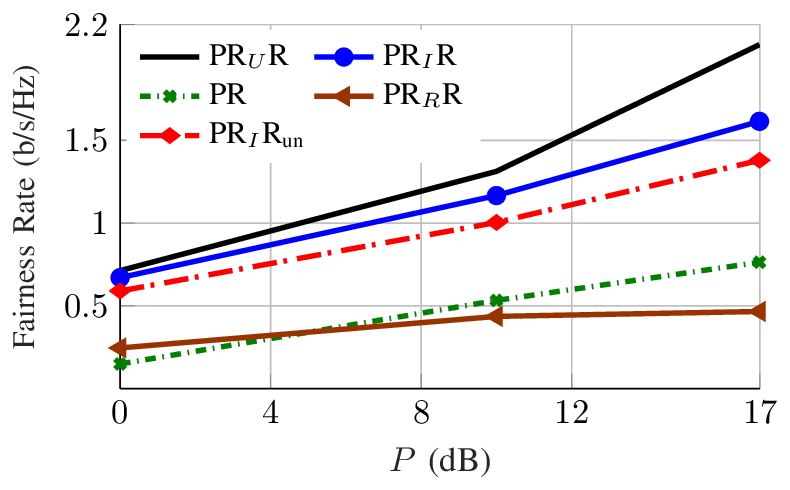}
        \caption{PGS schemes.}
    \end{subfigure}
	\\
	\begin{subfigure}[t]{0.5\textwidth}
        \centering
\includegraphics[width=.9\textwidth]{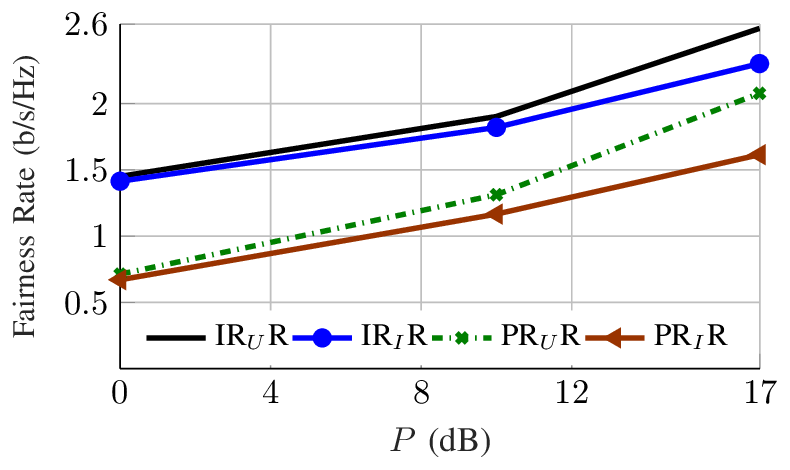}
        \caption{IGS and PGS schemes with RIS.}
    \end{subfigure}
    \caption{The average fairness rate versus $P$ for $N_{BS}=N_u=1$, $N_{RIS}=40$, $L=2$, $K=2$, $M=2$.}
	\label{Fig-rr}
\end{figure}

\subsection{Maximization of the minimum rate}
In this subsection, we consider the performance  of RS in max-min rate problems. We call the solution of MWRM problem as fairness rate since it usually results in the same rate for all the users \cite{soleymani2022improper, soleymani2022noma}.  
In the following, we first compare different 1-layer RS schemes with PGS/IGS with/without RIS. 
Then, we compare the 1-layer RS IGS/PGS schemes with TIN and the proposed algorithms in \cite{soleymani2022noma} for RIS-assisted systems. 
Finally, we compare the performance STAR-RIS with a regular RIS and traditional systems without RIS for the proposed 1-layer RS with IGS.

\subsubsection{Performance of RS}
Fig. \ref{Fig-rr} shows the average fairness rate versus $P$ for $N_{BS}=N_u=1$, $N_{RIS}=40$, $L=2$, $K=2$, $M=2$. As can be observed, RIS can considerably increase the fairness rate in both IGS and PGS schemes for all the feasibility sets when RIS components are properly optimized. Furthermore, we observe that IQI-aware schemes substantially outperform IQI-unaware PGS schemes. Fig. \ref{Fig-rr}c shows  that IGS with RS provides a much higher fairness rate than PGS schemes. We also observe that the schemes with the feasibility set $\mathcal{T}_I$  perform close to the upper bound especially when SNR is not high.   This is in line with our previous results in \cite{soleymani2022improper, soleymani2022noma}. Additionally, it can be observed that RIS with all four considered feasibility sets can considerably improve the system performance.

\begin{figure}[t!]
    \centering
\includegraphics[width=.5\textwidth]{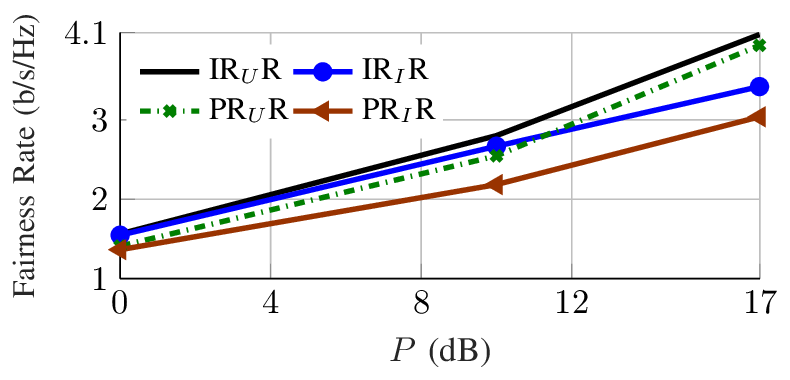}
    \caption{The average fairness rate versus $P$ for $N_{BS}=N_u=2$, $N_{RIS}=25$, $L=2$, $K=3$, $M=2$.}
	\label{Fig-rr-2} 
\end{figure}
Fig. \ref{Fig-rr-2} shows the average fairness rate versus $P$ for $N_{BS}=N_u=2$, $N_{RIS}=25$, $L=2$, $K=3$, $M=2$. As can be observed, IGS with RS outperforms the PGS scheme with RS; however, the benefits of IGS are less than in the SISO system shown  in Fig. \ref{Fig-rr}. The reason is that the number of spatial resources increases with $N_{BS}$ and $N_u$, which in turn reduces the interference level. Thus, the benefits of IGS as an interference-management tool reduce as well.

\begin{figure}[t!]
    \centering
    \begin{subfigure}[t]{0.5\textwidth}
        \centering
\includegraphics[width=.9\textwidth]{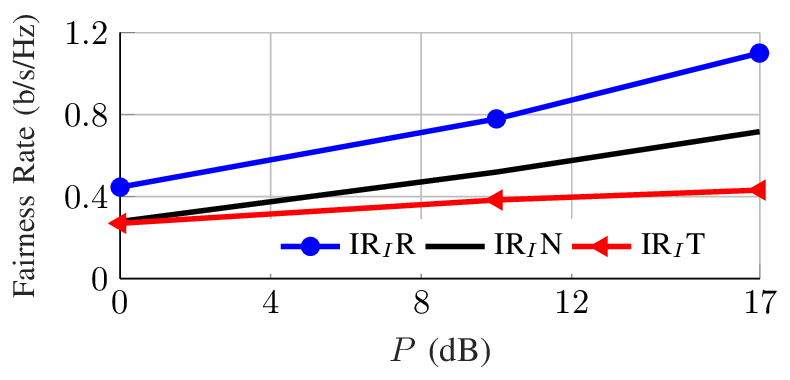}
        \caption{IGS schemes.}
    \end{subfigure}%
    \\ 
    \begin{subfigure}[t]{0.5\textwidth}
        \centering
\includegraphics[width=.9\textwidth]{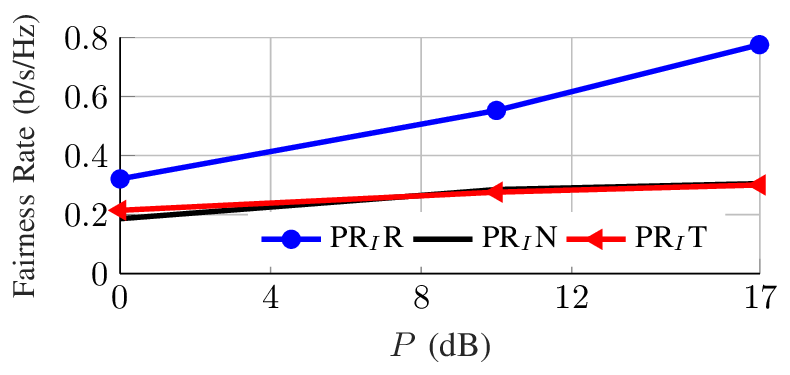}
        \caption{PGS schemes.}
    \end{subfigure}
	\\
	\begin{subfigure}[t]{0.5\textwidth}
        \centering
\includegraphics[width=.9\textwidth]{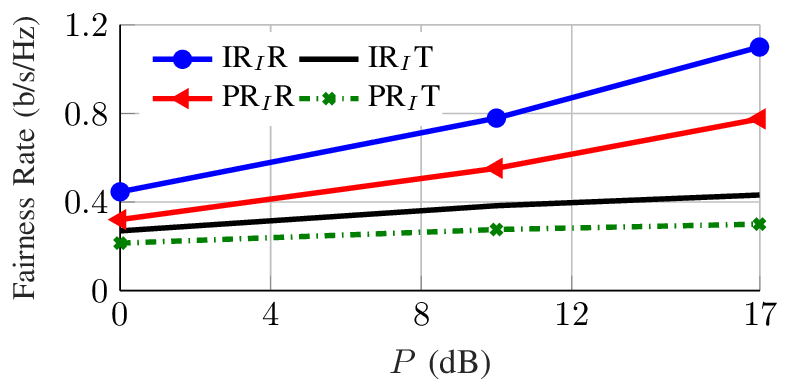}
        \caption{IGS and PGS schemes with RIS.}
    \end{subfigure}
    \caption{The average fairness rate versus $P$ for $N_{BS}=N_u=1$, $N_{RIS}=20$, $L=2$, $K=4$, $M=2$.}
	\label{Fig-rr-4} 
\end{figure}
\subsubsection{Comparison with TIN and/or NOMA-based schemes in \cite{soleymani2022noma}}
Fig. \ref{Fig-rr-4} shows the average fairness rate versus $P$ for $N_{BS}=N_u=1$, $N_{RIS}=20$, $L=2$, $K=4$, $M=2$. 
As can be observed, RS can significantly outperform TIN as well as the suboptimal NOMA-based schemes in \cite{soleymani2022noma}. It is worth emphasizing that {\em NOMA is the optimal scheme in a single-cell SISO BC with PGS and ideal devices, and RS cannot outperform NOMA in this set up}. However, it is not necessarily the case when we employ IGS and/or there is IQI at devices.  
Note that NOMA requires the optimal user ordering, which can be a difficult task in RIS-assisted systems, especially with multiple-antenna systems. Fig. \ref{Fig-rr-4}c shows that RS can significantly outperform TIN in both IGS and PGS schemes. Additionally, IGS with RS performs better than the other schemes.

\begin{figure}
        \centering
\includegraphics[width=.5\textwidth]{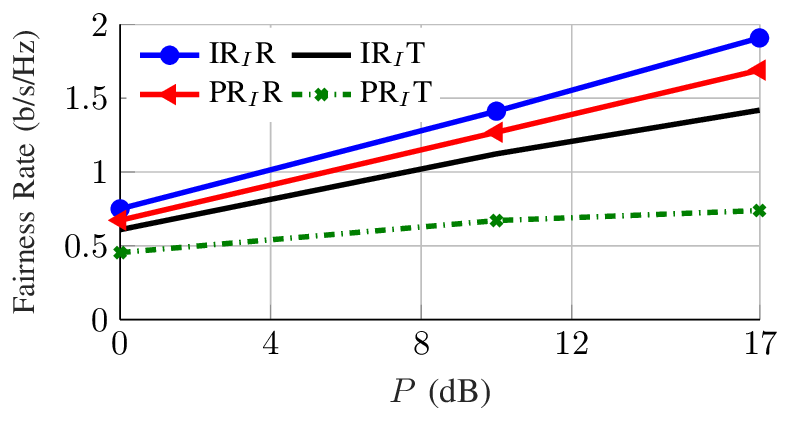}
    \caption{The average fairness rate versus $P$ for $N_{BS}=2$, $N_u=1$, $N_{RIS}=20$, $L=2$, $K=4$, $M=2$.}
	\label{Fig-rr-3} 
\end{figure}
Fig. \ref{Fig-rr-3} shows the average fairness rate versus $P$ for $N_{BS}=2$, $N_u=1$, $N_{RIS}=20$, $L=2$, $K=4$, $M=2$. As can be observed, the overall performance is similar to Fig. \ref{Fig-rr-4}. It means that RS can improve the system performance for both PGS and IGS schemes. Moreover, IGS with RS can outperform the other schemes. We also observe that the benefits of RS are much higher for PGS schemes. Indeed, since IGS is able to manage  part of interference, the benefits of RS are a bit lower, but still significant, in IGS schemes. 

\begin{figure}[t!]
    \centering
    \begin{subfigure}[t]{0.45\textwidth}
        \centering
\includegraphics[width=\textwidth]{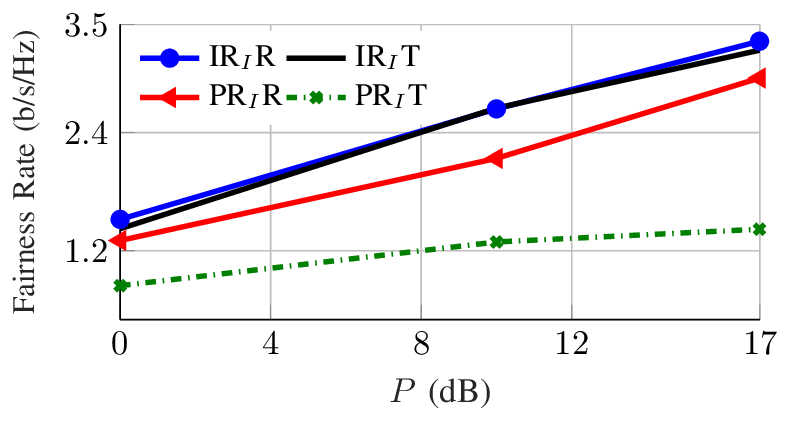}
        \caption{$K=3$.}
    \end{subfigure}%
    \\
    \begin{subfigure}[t]{0.45\textwidth}
        \centering
\includegraphics[width=\textwidth]{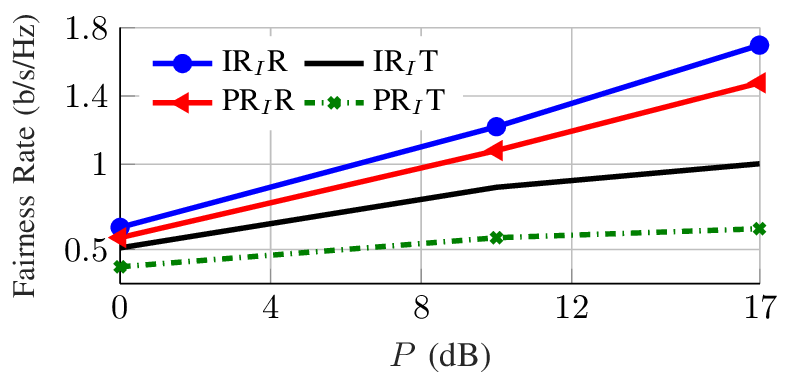}
        \caption{$K=5$.}
    \end{subfigure}
    \caption{The average fairness rate versus $P$ for $N_{BS}=N_u=2$, $N_{RIS}=15$, $L=2$,  $M=2$.}
	\label{Fig-rr2} 
\end{figure}
Fig. \ref{Fig-rr2} shows the average fairness rate versus the BS power budget  for $N_{BS}=N_u=2$, $N_{RIS}=15$, $L=2$,  $M=2$ with different number of users per cell. 
As can be observed, IGS brings a significant improvement over PGS in both RS and TIN schemes. Additionally, RS provides a considerable benefit for PGS schemes. However, the benefits of RS for IGS schemes are less than the benefits for the PGS schemes. 
The reason is that part of interference is already mitigated  by IGS, and hence, the benefits of RS as an interference-management technique, decrease in turn. 
When $K<2\max\left(N_{BS},N_u\right)$, IGS can effectively manage the interference, and there are no additional benefits by RS over IGS. However, when $K$ increases, the interference-level significantly increases, and as a result,  more advanced interference-management techniques are needed, i.e., a combination of RS and IGS. 
When $K>2\max\left(N_{BS},N_u\right)$, RS can provide significant gains in both IGS and PGS schemes, which shows the superiority of  RS in overloaded systems.    

\begin{figure}[t!]
    \centering
    \begin{subfigure}[t]{0.5\textwidth}
        \centering
        \includegraphics[width=.6\textwidth]{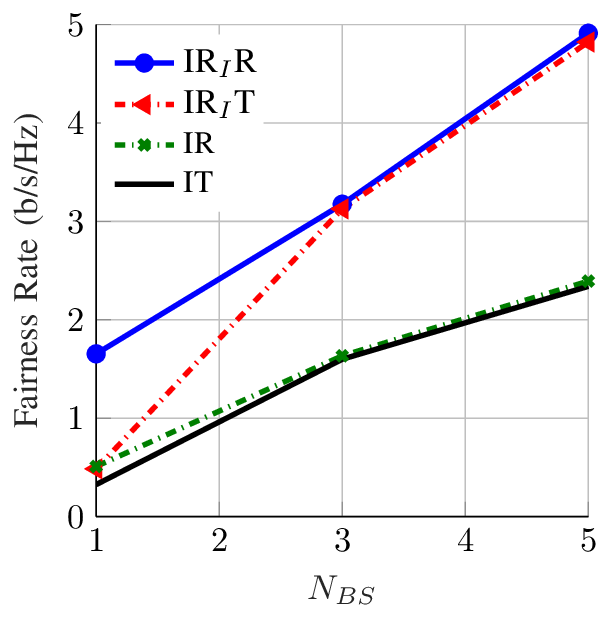}
        \caption{Average fairness rate.}
    \end{subfigure}%
    \\
    \begin{subfigure}[t]{0.5\textwidth}
        \centering
\includegraphics[width=.4\textwidth]{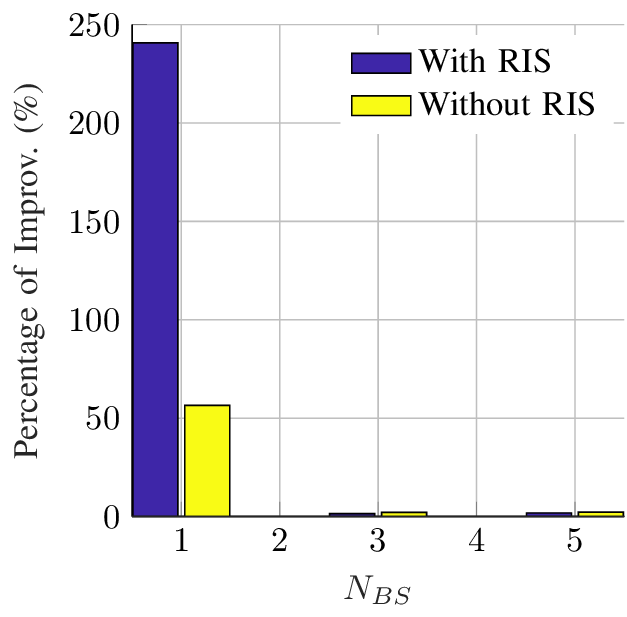}
        \caption{Benefits of RIS.}
    \end{subfigure}
    \caption{The average fairness rate versus $N_{BS}$ for $P=10$dB, $N_u=1$, $N_{RIS}=20$, $L=2$, and $M=2$.}
	\label{ris-rol}
\end{figure}
\subsubsection{Role of RIS} Fig. \ref{ris-rol} shows the average fairness rate versus $N_{BS}$  for $P=10$ dB, $N_u=1$, $N_{RIS}=20$, $L=2$, and $M=2$. As can be observed, the benefits of RS significantly decrease with $N_{BS}$ since the number of spatial resources per users increases, and the system changes from highly overloaded to underloaded as $N_{BS}$ grows. On the contrary, RIS provides a considerable gain for all $N_{BS}$ especially when it is combined with RS. In other words, RIS can improve the performance of both underloaded and overloaded systems, which is mainly due to the coverage improvements by RIS. Additionally, we observe that RIS can significantly increase the benefits of RS for $N_{BS}=1$ (and vice versa), which may indicate that RIS and RS are mutually beneficial tools  to improve performance in highly overloaded systems. Note that in our previous study \cite{soleymani2022noma}, we showed that RIS alone cannot completely  handle the undesired consequences of interference in a multicell BC. Therefore, to fully exploit the RIS benefits, other advanced interference-management techniques are necessary, which is in line with the results in Fig. \ref{ris-rol}.

\subsubsection{STAR-RIS-assisted systems} 
\begin{figure}
        \centering
\includegraphics[width=.5\textwidth]{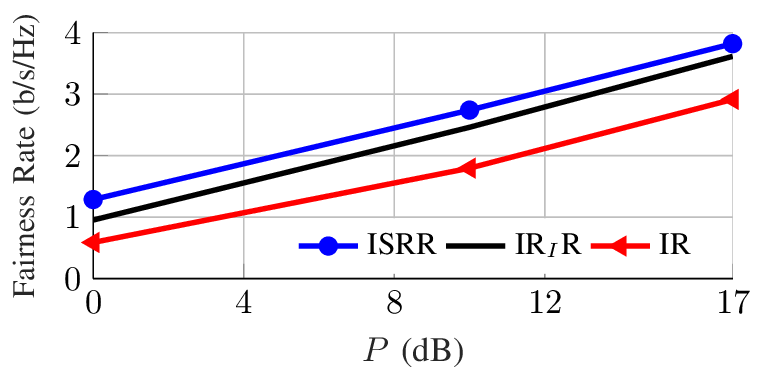}
    \caption{The average fairness rate versus $P$ for $N_{BS}=2$, $N_u=2$, $N_{RIS}=20$, $L=1$, $K=4$, and $M=1$.}
	\label{Fig-str-3} 
\end{figure}

We also consider STAR-RIS-assisted systems in the simulations. In the numerical results for STAR-RIS-assisted systems, we consider a single-cell BC with $N_{BS}=2$, $N_u=2$, $N_{RIS}=20$, $L=1$, $K=4$, and $M=1$.  We assume that the conventional RIS can assist only half of the users, while the two other users are blocked and do not receive any signal from the RIS. 
However, the STAR-RIS can cover all the users. Specifically, two users are in the transmission space of the STAR-RIS, and the other two users are in the reflection space of the STAR-RIS. 
For a fair comparison, we assume that both STAR-RIS and regular RIS have the same number of components. In this particular example, we consider the mode switching protocol for the STAR-RIS, which means that half of the components operate in the transmission mode, and the other half operates in the reflection mode. Note that there are also other protocols. However, due to a space restriction, we leave the full study of STAR-RIS for a future work. As can be observed in Fig. \ref{Fig-str-3}, STAR-RIS can outperform the regular RIS. Additionally, we observe that RIS (either regular or STAR) can drastically increase the fairness rate. 

\subsection{Minimization of the total transmit power for a target rate}
\begin{figure}[t!]
    \centering
    \begin{subfigure}[t]{0.45\textwidth}
        \centering
\includegraphics[width=\textwidth]{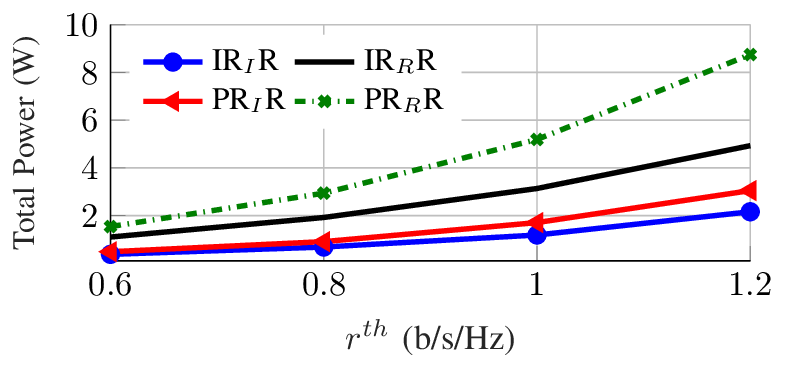}
        \caption{$K=3$.}
    \end{subfigure}%
    \\ 
    \begin{subfigure}[t]{0.45\textwidth}
        \centering
\includegraphics[width=\textwidth]{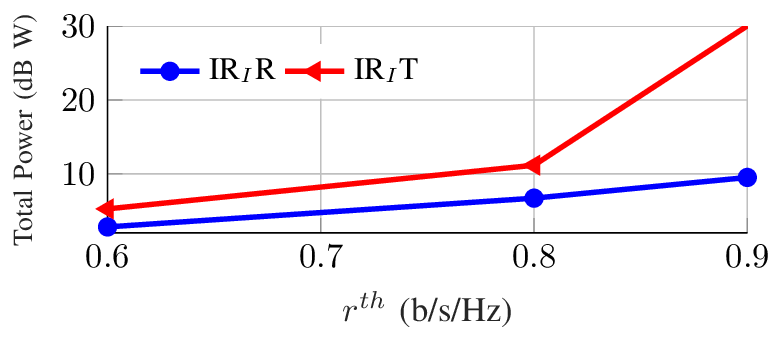}
        \caption{$K=5$.}
    \end{subfigure}
    \caption{The average minimum total transmission power versus the target rate $R_{th}$ for $N_{BS}=N_u=2$, $N_{RIS}=15$, $L=2$,  $M=2$.}
	\label{Fig-min-p}
\end{figure}
Fig. \ref{Fig-min-p} shows the  minimum average power versus the target rate $r^{th}$ for $N_{BS}=N_u=2$, $N_{RIS}=15$, $L=2$,  $M=2$. 
As can be observed, the total transmission power significantly increases with $r^{th}$. 
It should be noted that the target rate $r^{th}$ might not be necessarily achievable by a scheme for a specific channel, especially if we set a high threshold  $r^{th}$.
In Fig. \ref{Fig-min-p}, we observe that IGS with RIS and RS improves energy efficiency and substantially reduces the power consumption. Note that the scenario in Fig. \ref{Fig-min-p}b is a highly overloaded network. This figure shows that RS can be very energy efficient, especially when the threshold is set to high values. Note that the scale of the $y$-axis is in dB in Fig. \ref{Fig-min-p}b, and IGS with RS achieves 0.9 b/s/Hz with the total average power less than 10W, while IGS with TIN requires around 30dBW total power on average to reach the rate, which shows the superiority of RS over TIN from an energy efficiency point of view. 

\subsection{Maximization of the minimum EE}

\begin{figure}
        \centering
\includegraphics[width=.5\textwidth]{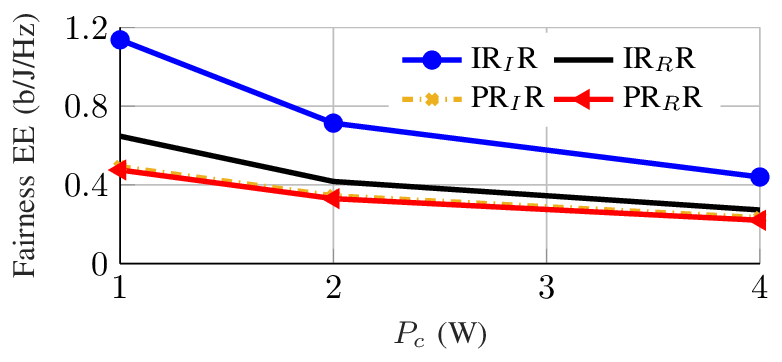}
        \caption{The average fairness EE versus $P$ for $N_{BS}=N_u=2$, $N_{RIS}=25$, $L=2$, $K=2$, $M=2$.}
	\label{Fig-ee} 
\end{figure}
In the previous subsections, we have considered several scenarios with different parameters.  It can be expected that we observe the same behavior with EE metrics. Thus, in this subsection, we consider only one scenario for the MWEEM problem. 
Fig. \ref{Fig-ee} shows the average fairness EE versus $P$ for $N_{BS}=N_u=2$, $N_{RIS}=25$, $L=2$, $K=2$, $M=2$.
As can be observed, IGS with RS outperforms the PGS scheme with RS. Moreover, it is clear that the IGS scheme with random reflecting coefficients does not provide any gain over PGS, which shows the importance of optimizing the reflecting coefficients.

\subsection{Summary}
Our main findings in the numerical section can be summarized as follows:
\begin{itemize}
\item IGS schemes always improve the system performance over PGS schemes with both RS and TIN because of two reasons. First, IGS can compensate IQI. Second, IGS is in addition an efficient interference-management technique.

\item RS when combined with IGS may substantially improve the system performance with different performance metrics specially when the system is highly overloaded. A metric for measuring the overload of a system can be defined as the comparison between the number of users per cell and the maximum number of transmit/receive antennas. If $K>\max(N_{BS},N_u)$, the system is considered as overloaded.

\item The benefits of RS are higher in PGS schemes than in IGS schemes. The reason is that part of the  interference can be managed by IGS, which reduces the effective interference level, and in turn, yields less benefits by RS as an interference-management technique. 

\item The more overloaded a system is, the more benefits by the RS in combination with IGS.  In other words, the benefits of the RS with IGS scheme increase with the number of users per cell and decrease with the number of transmit/receive antennas.
\end{itemize}

\section{Conclusion and future work}\label{sec-con}
In this paper, we have proposed a general framework for RS in MIMO RIS-assisted systems. 
The optimization framework yields a stationary point of every optimization problem in which the objective and/or constraints are linear functions of rates and/or transmit covariance matrices when the feasibility set of the RIS components is convex. 
These optimization problems include, e.g., weighted-minimum and weighted-sum rate maximization, total transmit power minimization, weighted-minimum EE and global EE maximization.  Additionally, this framework accounts for hardware imperfections such as IQI. 
As an illustrative example, we considered a multicell MIMO RIS-assisted BC with IQI at all BSs and users, showing that RS can substantially improve the spectral and energy efficiencies of this system when the number of users per cell is higher than the maximum number of transmit and receive antennas.

As a future work, the performance of RS and IGS should be investigated in the presence of imperfect and/or statistical CSI at transceivers. Furthermore, another challenging future research line can be to consider multiple layer RS schemes.

\appendices
\section{IQI model and improper signals}\label{sec-iqi}
When there is an imbalance in in-phase and quadrature components of a device, the output signal is improper. In a zero-mean improper signal, the real and imaginary parts are correlated and/or have unequal powers \cite{schreier2010statistical, adali2011complex}.
Consider a zero-mean Gaussian random vector $\mathbf{x}$. It is called proper if its complementary variance is equal to zero, i.e., $\mathbb{E}\{\mathbf{x}\mathbf{x}^T\}=\mathbf{0}$ \cite{schreier2010statistical, adali2011complex}. Otherwise, $\mathbf{x}$ is called improper \cite{schreier2010statistical, adali2011complex}. 
To model impropriety, there are different approaches. In this paper, we employ the real-decomposition method since it is more convenient to optimize over the rate expression \cite{soleymani2020improper, soleymani2022improper}. We refer the readers to \cite[Sec. II.A]{soleymani2022improper} for more details on the real-decomposition method to model impropriety in MIMO systems.

IQI can be modeled as a widely linear transformation \cite{javed2019multiple, soleymani2020improper}, which means that the output signal is a linear transformation of the input signal and its conjugate \cite{schreier2010statistical, adali2011complex}.  
In this paper, we employ the IQI model  for MIMO systems in \cite{javed2019multiple}, which was later used in \cite{soleymani2020improper, soleymani2022improper}. For the sake of completeness, we briefly restate the model here and refer the reader to \cite{javed2019multiple, soleymani2020improper, soleymani2022improper} for more details. 
Consider a typical point-to-point MIMO system with $N_t$ transmit antennas and $N_r$ receive antennas with IQI at the transceivers. 
The received signal can be written as \cite[Eq. (6)]{soleymani2022improper}
\begin{multline}\label{eq-1}
\mathbf{y}=\mathbf{\Gamma}_{r,1}\left[\mathbf{H}\left(\mathbf{\Gamma}_{t,1}\mathbf{x}+\mathbf{\Gamma}_{t,2}\mathbf{x}^*\right)+\mathbf{r}\right]
\\
+\mathbf{\Gamma}_{r,2}\left[\mathbf{H}\left(\mathbf{\Gamma}_{t,1}\mathbf{x}+\mathbf{\Gamma}_{t,2}\mathbf{x}^*\right)+\mathbf{r}\right]^*,
\end{multline}
where the parameters are defined in \cite[Eqs. (7)-(10)]{soleymani2022improper}.
The transmitter (receiver) is ideal if $\mathbf{\Gamma}_{t,1}=\mathbf{I}$ ($\mathbf{\Gamma}_{r,1}=\mathbf{I}$) and $\mathbf{\Gamma}_{t,2}=\mathbf{0}$ ($\mathbf{\Gamma}_{r,2}=\mathbf{0}$). In the following lemma, we represent the real decomposition of \eqref{eq-1}.

\begin{lemma}[\!\!\cite{soleymani2022improper}]\label{lem:Rea}
Employing the real-decomposition method, the point-to-point MIMO system with IQI can be modeled as $\underline{\mathbf{y}}=
\underline{\mathbf{H}}\,
\underline{\mathbf{x}}+\underline{\mathbf{n}},$
where $\underline{\mathbf{y}}=\left[ \begin{array}{cc}
\mathfrak{R}\{\mathbf{y}\}^T & \mathfrak{I}\{\mathbf{y}\}^T \end{array} \right]^T$,  
$\underline{\mathbf{x}}=\left[ \begin{array}{cc}
\mathfrak{R}\{\mathbf{x}\}^T & \mathfrak{I}\{\mathbf{x}\}^T \end{array} \right]^T$, and 
$\underline{\mathbf{n}}=\left[ \begin{array}{cc}
\mathfrak{R}\{\mathbf{n}\}^T & \mathfrak{I}\{\mathbf{n}\}^T \end{array} \right]^T$ are, respectively, the real decomposition of $\mathbf{y}$, $\mathbf{x}$, and $\mathbf{n}=\mathbf{\Gamma}_{r,1}\mathbf{r}+\mathbf{\Gamma}_{r,2}\mathbf{r}^*$. 
Moreover, $\underline{\mathbf{H}}$ is the equivalent channel, given by \cite[Eq. (11)]{soleymani2022improper}.
The statistics of the vector $\underline{\mathbf{n}}\in \mathbb{R}^{2N_r\times 1}$ are $\mathbb{E}\{\underline{\mathbf{n}}\}=\mathbf{0}$, and
$\mathbb{E}\{\underline{\mathbf{n}}\,\underline{\mathbf{n}}^T\}=\underline{\mathbf{C}}_{n}= \underline{\mathbf{\Gamma}}\,
\underline{\mathbf{C}}_r
\underline{\mathbf{\Gamma}}^T,$
where $\underline{\mathbf{C}}_r$ is the real decomposition of $\mathbf{C}_r$ and $\underline{\mathbf{\Gamma}}$ is given by \cite[Eq. (13)]{soleymani2022improper}. 
\end{lemma}

\section{Preliminaries on majorization minimization}\label{ap=MM}
Consider the following optimization problem
\begin{align}\label{ar-opt-jadid}
 \underset{\{\mathbf{P}\}\in\mathcal{P}
 }{\max}\,\,  & 
  f_0\left(\{\mathbf{P}\}\right) &
 \,\,\,\,\, \,\, \text{s.t.}   \,\,\,\,\,&  f_i\left(\{\mathbf{P}\}\right)\geq0,&\forall i,
 \end{align}
where $\{\mathbf{P}\}$ is the set of optimization parameters, and $\mathcal{P}$ is the feasibility set of the variables.
If $f_i$s for all $i$ are concave, and $\mathcal{P}$ is a convex set, the optimization problem \eqref{ar-opt-jadid} is known as convex   and can be solved in polynomial time \cite{boyd2004convex}. Unfortunately, it is not the case in most of practical systems, especially with RS and RIS. Thus, we resort to some optimization techniques to efficiently solve \eqref{ar-opt-jadid}. A powerful numerical optimization tool is majorization minimization (MM), which includes many iterative optimization techniques such as expectation-maximization (EM), sequential convex programming (SCP) and difference of convex programming (DCP) \cite{sun2017majorization}. In the following, we provide a brief review on the main idea of MM and refer the reader to \cite{sun2017majorization} for a more detailed overview of MM-based techniques.

\begin{figure}[t]
\centering
\includegraphics[width=.43\textwidth]{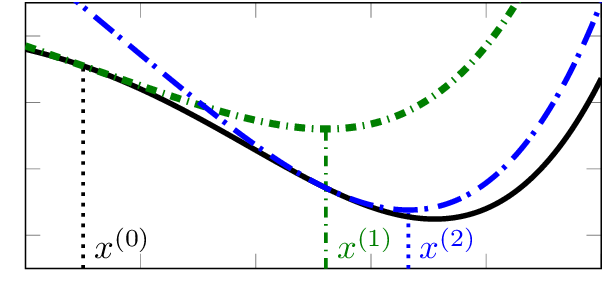}
\caption{An example of Majorization minimization.}
\label{MM-approach-fig}
\end{figure}
MM is an iterative optimization approach, which consists of two steps in each iteration: majorization and minimization. 
In the majorization step of the $t$-th iteration, the non-concave objective and/or constraint function $f_i$ is approximated by a suitable surrogate function $\tilde{f}_i^{(t)}$ that yields  a surrogate optimization problem. Then, the surrogate problem  is solved in the minimization step, which gives a new point as depicted in Fig. \ref{MM-approach-fig}. 
This procedure is iterated until a convergence metric is met. 
MM converges to a stationary point of the original problem. 
In MM, the surrogate functions are not necessarily concave, but they should fulfill three specific conditions shown in the following lemma.
\begin{lemma}[\!\!\cite{aubry2018new}]\label{lem-sur-jadid}
Assume that $\tilde{f}_i^{(t)}$, which is a surrogate function for $f_i$ at the $t$-th iteration, fulfills the following conditions:
\begin{itemize}
\item $\tilde{f}_i^{(t)}\left(\{\mathbf{P}^{(t-1)}\}\right)=f_i\left(\{\mathbf{P}^{(t-1)}\}\right)$. 
\item $\frac{\partial \tilde{f}_i^{(t)}\left(\{\mathbf{P}^{(t-1)}\}\right)}{\partial \mathbf{P}_k}
=\frac{\partial f_{i}\left(\{\mathbf{P}^{(t-1)}\}\right)}{\partial \mathbf{P}_k}$ for all $k$.
\item $\tilde{f}_i^{(t)}\leq f_i$ for the whole domain,
\end{itemize}
where $\{\mathbf{P}^{(t-1)}\}$ is the initial point at the $t$-th iteration, given by 
\begin{align}\label{ar-opt-2-jadid-2}
 \underset{\{\mathbf{P}\}\in\mathcal{P}
 }{\max}\,\,\,\,\,\,\,\,  & 
  \tilde{f}_0^{(t-1)} &
 \,\,\,\,\,\,\,\, \,\,\,\,\,\,\, \text{{\em s.t.}}  \,\,\,\,\,\,\,\,\,&  \tilde{f}_i^{(t-1)}\geq0,&\forall i.
 \end{align}
Then, the sequence of $\{\mathbf{P}^{(t)}\}$ converges to a stationary point of \eqref{ar-opt-jadid}. 
\end{lemma}
The most difficult task in MM algorithms is to obtain suitable surrogate functions. In wireless communications, we often deal with achievable rate, which is a logarithmic function for Gaussian signals with infinite block lengths \cite{cover2012elements}. Rates are mostly continuous-valued and differentially continuous functions, which makes the majorization task simpler. For instance, rates are either concave, convex and/or a difference of two concave functions (depending on the scenario) in transmit covariance matrices. Hence, 
it is customary to employ
the first-order Taylor expansion and the convex-concave procedure (CCP) to find suitable surrogate functions. CCP is based on the  following inequality \cite[Eq. (15)]{sun2017majorization}
\begin{equation}\label{eq04}
f_{cvc}\leq f_{l}\leq f_{cvx},
\end{equation}
where $f_{cvc}$, $f_{l}$ and $f_{cvx}$ are, respectively, concave, linear/affine and convex functions with an equal value and an equal first order derivative at an arbitrary feasible point $\mathbf{P}^{(t)}$. In the following lemmas, we state some frequently used inequalities, which are based on \eqref{eq04}. We refer the readers to \cite[Lemma 1]{soleymani2022noma} and \cite{yu2020improper, yu2020joint} for the proofs and/or further details. 
\begin{lemma}
\label{lem-1} 
Consider $f(\{\mathbf{P}\})=\ln\left|\mathbf{A}+\sum_{i=1}^I\mathbf{B}_i\mathbf{P}_i\mathbf{B}^T_i\right|,$ where 
$\mathbf{A}\in\mathbb{R}^{N\times N}$ and $\mathbf{B}_i\in\mathbb{R}^{N\times M}$ for all $i$ are constant matrices. Additionally, $\mathbf{A}$ and $\mathbf{P}_i\in\mathbb{R}^{M\times M}$ for all $i$ are  positive semi-definite matrices. 
%
Then, we have following inequality for all feasible $\{\mathbf{P}\}$ 
\begin{multline*}
f(\{\mathbf{P}\}) \leq f\left(\left\{\mathbf{P}^{(t)}\right\}\right) \\
 +\!\sum_{i=1}^I\! \text{\em{Tr}}\!\left(\!\mathbf{B}^T_i\left(\!\mathbf{A}+\!\!\sum_{i=1}^I\mathbf{B}_i\mathbf{P}^{(t)}_i\mathbf{B}^T_i\!\right)^{-1}\!\!\mathbf{B}_i(\mathbf{P}_i-\mathbf{P}^{(t)}_i)\!\right)\!, 
\end{multline*}
where $\left\{\mathbf{P}^{(t)}\right\}=\left\{\mathbf{P}^{(t)}_1,\mathbf{P}^{(t)}_2,\cdots,\mathbf{P}^{(t)}_I\right\}$ is any feasible fixed point.
\end{lemma}
\begin{proof}
Since $f(\{\mathbf{P}\})$ is concave in $\{\mathbf{P}\}$, we can obtain the following inequality by employing the first-order Taylor expansion
\begin{equation*}
f(\mathbf{P}) \leq f(\mathbf{P}^{(t)})  +\sum_i \text{{Tr}}\left(\left[\frac{\partial f(\{\mathbf{P}\})}{\partial \mathbf{P}_i}|_{\{\mathbf{P}^{(t)}\}}\right]^T(\mathbf{P}_i-\mathbf{P}^{(t)}_i)\!\right)\!, 
\end{equation*}
where $\frac{\partial f(\{\mathbf{P}\})}{\partial \mathbf{P}_i}|_{\{\mathbf{P}^{(t)}\}}$ is the derivative of $f(\{\mathbf{P}\})$ with respect to $\mathbf{P}_i$ at $\{\mathbf{P}^{(t)}\}$, which is given by
\begin{equation*}
\frac{\partial f(\{\mathbf{P}\})}{\partial \mathbf{P}}|_{\mathbf{P}^{(t)}}=\mathbf{B}^T_i\left(\!\mathbf{A}+\!\!\sum_{i=1}^I\mathbf{B}_i\mathbf{P}^{(t)}_i\mathbf{B}^T_i\!\right)^{-1}\!\!\mathbf{B}_i.
\end{equation*}
It is easy to verify that $\frac{\partial f(\mathbf{P})}{\partial \mathbf{P}}$ is symmetric, which proves the lemma.
\end{proof}
\begin{lemma}[\!\!\cite{yu2020improper, yu2020joint}]\label{lem-2} 
The following inequality holds for all  $N \times N$ Hermitian positive definite matrices $\mathbf{Y}$ and $\bar{\mathbf{Y}}$, and any arbitrary  $N \times M$ complex  matrices $\mathbf{V}$ and $\bar{\mathbf{V}}$:
\begin{multline*}
\ln \left|\mathbf{I}+\mathbf{V}\mathbf{V}^H\mathbf{Y}^{-1}\right|
\geq
 \ln \left|\mathbf{I}+\bar{\mathbf{V}}\bar{\mathbf{V}}^H\bar{\mathbf{Y}}^{-1}\right|
 \\
 -
\text{{\em Tr}}\left(
\bar{\mathbf{V}}\bar{\mathbf{V}}^H\bar{\mathbf{Y}}^{-1}
\right)
+
2\mathfrak{R}\left\{\text{{\em Tr}}\left(
\bar{\mathbf{V}}^H\bar{\mathbf{Y}}^{-1}\mathbf{V}
\right)\right\}\\
-
\text{{\em Tr}}\left(
(\bar{\mathbf{Y}}^{-1}-(\bar{\mathbf{V}}\bar{\mathbf{V}}^H + \bar{\mathbf{Y}})^{-1})^H(\mathbf{V}\mathbf{V}^H+\mathbf{Y})
\right).
\end{multline*}
\end{lemma}
\begin{lemma}\label{lem=5}
The following inequality holds for all $t_i$
\begin{equation}
\sum_{i=1}^I|t_i|^2\geq \sum_{i=1}^I|t_i^{(n)}|^2+2\sum_{i=1}^I\mathfrak{R}\left\{t_i^{(n)^*}(t_i-t_i^{(n)})\right\},
\end{equation}
where $t_i^{(n)}$ is any arbitrary point. 
\end{lemma}
\begin{proof}The function $|t_i|^2$ is convex for all $i$. Thus, we can apply the the first-order Taylor expansion and the inequality to obtain a linear lower bound for $|t_i|^2$.
\end{proof}

\section{Preliminaries on RIS}\label{sec=ap=ris}
In this appendix, we provide a short preliminary discussion  on the feasibility set of RIS components and refer the reader to \cite{pan2020multicell,soleymani2022improper} for more details. 
Additionally, we briefly present the main idea of STAR-RISs.

\subsection{Feasibility sets for the RIS components} \label{sec=ap=ris-c}

We consider the three feasibility sets in \cite{soleymani2022noma} throughout this paper.
The ideal feasibility set for the reflecting coefficients is \cite[Eq. (11)]{wu2021intelligent}
\begin{equation}
\mathcal{T}_{U}=\left\{\theta_{m_n}:|\theta_{m_n}|^2\leq 1 \,\,\,\forall m,n\right\},
\end{equation}
in which the amplitude and the phase of each component can be independently optimized  \cite{wu2021intelligent, elmossallamy2020reconfigurable, yang2020risofdm}.
This feasibility set is a convex set, but it is not practical \cite{di2020smart, wu2021intelligent}, although it serves to assess  the theoretical performance limit for RIS-assisted systems \cite{wu2021intelligent}. 
 A more practical feasibility set assumes that the amplitudes are fixed to $1$, and only the phases can be optimized, i.e.,  
\begin{equation}
\mathcal{T}_{I}=\left\{\theta_{mi}:|\theta_{mi}|= 1 \,\,\,\forall m,i\right\}.
\end{equation}
This feasibility set is very popular and has been used in many works such as \cite{di2020smart, wu2021intelligent, wu2019intelligent, kammoun2020asymptotic, yu2020joint, pan2020multicell, zhang2020intelligent}.
Another practical feasibility set assumes that the amplitude  of each RIS element is not fixed, but it is a deterministic function of its phase \cite{abeywickrama2020intelligent}. In \cite{abeywickrama2020intelligent}, the amplitude as a function of the phase is 
\begin{equation}\label{eq*=*}
\mathcal{F}(\angle \theta_{mi})= |\theta|_{\min}+( 1\!-|\theta|_{\min})\left(\!\!\frac{\sin\left(\angle \theta_{mi}-\phi\right)+1}{2}\!\right)^{\alpha}\!\!\!,
\end{equation}
where $|\theta|_{\min}$, $\alpha$, and $\phi$ are non-negative constant values.
This feasibility set can be formulated as  \cite[Eq. (7)]{soleymani2022noma}
\begin{equation}
\mathcal{T}_{C}\!=\!\left\{\theta_{mi}\!:|\theta_{mi}|= \mathcal{F}(\angle \theta_{mi}), \,\angle \theta_{mi}\in[-\pi,\pi]  \,\forall m,i\right\}\!.
\end{equation}
Finally, we consider a feasibility set in which the phases of the RIS components are discrete while the amplitude is set to 1 as 
\begin{equation}
\mathcal{T}_{D}=\left\{\theta_{mi}:|\theta_{mi}|= 1,\angle\theta_{mi}=\{\phi_1,\phi_2,\cdots,\phi_N\} \,\,\,\forall m,i\right\},
\end{equation}
where $\phi_n$s for all $n$ are the only possible phase shifts that can be tuned \cite{wu2021intelligent}.

\subsection{STAR-RIS}\label{sec=ap=ris-d}
In STAR-RIS, each RIS component can simultaneously reflect and transmit signals, which provides an additional degree of freedom in the design. The reflection and transmission components of the $i$-th element of the $m$-th RIS are denoted, respectively, by $\theta_{m_i}^{r}$ and $\theta_{m_i}^{t}$.
They are constrained as $|\theta_{m_i}^{r}|\leq 1$ and $|\theta_{m_i}^{t}|\leq 1$ \cite{mu2021simultaneously, wu2021coverage, liu2021star, xu2021star}. 
Unfortunately, $\theta_{m_i}^{r}$ and $\theta_{m_i}^{t}$ are related to each other, and their values cannot be optimized independently. There are two different models that account for  this relationship. The most common model is  \cite{wu2021coverage}, \cite[Eq. (1)]{liu2021star}
\begin{equation}\label{model-1}
|\theta_{m_i}^{r}|^2+|\theta_{m_i}^{t}|^2=1.
\end{equation}
An alternative model is \cite[Eq. (2)]{xu2021star}
\begin{equation}\label{model-2}
|\theta_{m_i}^{r}|^2+|\theta_{m_i}^{t}|^2\leq1.
\end{equation}
Clearly, \eqref{model-2} is a convex constraint, while \eqref{model-1} is not,
 which makes the optimization problems with this constraint more complicated.
 
In STAR-RIS-assisted systems, there are two spaces for each RIS: reflection space (RS) and transmission space (TS), and each user belongs to either of them \cite{mu2021simultaneously}. Thus, the channel for a user is \cite[Eq. (2)]{mu2021simultaneously}  
\begin{equation*}
\mathbf{H}_{lk,i}\left(\{\bm{\Theta}\}\right)=
\underbrace{\sum_{m=1}^M\mathbf{G}_{lk,m}\bm{\Theta}_m^{r/t}\mathbf{G}_{m,i}}_{\text{Link through RIS}}
+
\underbrace{\mathbf{F}_{lk,i}}_{\text{Direct link}}\in\mathbb{C}^{N_u\times N_{BS}},
\end{equation*}
where
\begin{align*}
\bm{\Theta}_m^r&
=\text{diag}\left(\theta_{m_1}^r, \theta_{m_2}^r,\cdots,\theta_{m_{N_{RIS}}}^r\right), \\
\bm{\Theta}_m^t&
=\text{diag}\left(\theta_{m_1}^t, \theta_{m_2}^t,\cdots,\theta_{m_{N_{RIS}}}^t\right)
\end{align*}

Three different operational modes are proposed in \cite{mu2021simultaneously} for STAR-RIS: energy splitting (ES), mode switching (MS), and time sharing (TS). In ES, all RIS components can simultaneously transmit and reflect. In MS, the RIS components can either transmit or reflect. In other words, in this mode, the RIS components are divided into two separate groups. In TS, all the RIS components periodically switch between transmission and reflection modes in different orthogonal time slots.  
Note that MS and TS are special cases of ES. Thus, in this work, we consider only ES without loss of generality.

\bibliographystyle{IEEEtran}
\bibliography{ref2}

\end{document}